\documentclass[a4paper, accepted=2022-06-11]{quantumarticle}
\pdfoutput=1
\usepackage{url}
\usepackage{authblk}
\usepackage[utf8]{inputenc}
\usepackage{amsmath}
\usepackage{amsfonts}
\usepackage{amssymb}
\usepackage{amsthm}
\usepackage{graphicx}
\usepackage{mathtools,nccmath}
\usepackage{xcolor}
\usepackage{geometry}
\usepackage{dsfont}
\usepackage{physics}

\newtheorem{lemma}{Lemma}
\newtheorem{theorem}{Theorem}
\newtheorem{proposition}{Proposition}
\newtheorem{definition}{Definition}
\newtheorem{assumption}{Assumption}

\newgeometry{vmargin={25mm}, hmargin={14mm,18mm}}   

\newcommand{\diam}{\mathrm{diam}}
\newcommand{\dist}{\mathrm{dist}}
\newcommand{\Ker}{\mathrm{Ker}}
\newcommand{\caH}{\mathcal{H}}
\newcommand{\caG}{\mathcal{G}}
\newcommand{\caA}{\mathcal{A}}
\newcommand{\caB}{\mathcal{B}}
\newcommand{\caM}{\mathcal{M}}

\newcommand{\superadd}{\mathrm{S}}
\newcommand{\bbN}{\mathbb{N}}
\newcommand{\bbR}{\mathbb{R}}
\newcommand{\bbC}{\mathbb{C}}
\newcommand{\aux}{^{\mathrm{aux}}}
\newcommand{\spec}{{\mathrm{spec}}}

\newcommand{\ident}{\mathds{1} }
\newcommand{\bbE}{\mathbb{E}}

\title{Stability of invertible, frustration-free ground states against large perturbations}
\author[1]{Sven Bachmann}
\author[2]{Wojciech De Roeck}
\author[3]{Brecht Donvil}
\author[4]{Martin Fraas}
\affil[1]{Department of Mathematics, University of British Columbia, Vancouver, BC V6T 1Z2, Canada}
\affil[2] {Institute of Theoretical Physics, K.U. Leuven, 3001 Leuven, Belgium }
\affil[3]{Institute for Complex Quantum Systems and Center for IQST, Ulm University, 89069 Ulm, Germany}
\affil[3]{Department of Mathematics and Statistics, University of Helsinki, Helsinki, Finland}

\affil[4]{Department of Mathematics, University of California, Davis, Davis, CA, 95616, USA}

\begin{document}
\maketitle
\abstract{A gapped ground state of a quantum spin system has a natural length scale set by the gap. This length scale governs the decay of correlations.  A common intuition is that this length scale also controls the spatial relaxation towards the ground state away from impurities or boundaries. The aim of this article is to take a step towards a proof of this intuition. We assume that the ground state is frustration-free and invertible, i.e.\ it has no long-range entanglement. Moreover, we assume the property that we are aiming to prove for one specific kind of boundary condition; namely open boundary conditions. This assumption is also known as the "local topological quantum order" (LTQO) condition.   With these assumptions we can prove stretched exponential decay away from  boundaries or impurities, for any of the ground states of the perturbed system. In contrast to most earlier results, we do not assume that the perturbations at the boundary or the impurity are small. In particular, the perturbed system itself can have long-range entanglement. 
}
\section{Informal statement of the result}\label{sec: informal statement on stability}
Since a full ab initio statement of our assumptions and theorem requires quite some setup and definitions, we first state some simplified assumptions and the result, freely using terminology that is probably known to most of our readers.
We consider a spin system on a finite discrete set $\Gamma$, say, a subset of $\mathbb{Z}^d$. We define a  Hamiltonian $H+J$ where both terms $H$ and $J$ are local Hamiltonians, i.e.\ sums of local terms $H=\sum_{X \subset \Gamma} h_X, J=\sum_{X \subset \Gamma} j_X$ with $||h_X||,||j_X||$ decaying rapidly in $\diam(X)$.  We assume moreover the following properties: 
\begin{enumerate}
\item The spatial support of $J$ is confined to  a region $\Gamma_j$  with arbitrary size.  Crucially, we do not assume that the terms $j_X$  are small. 
\item The Hamiltonian $H$ has a ground state $\Omega$  that is \emph{invertible}. This means there is an auxiliary state $\Omega'$ such that $\Omega\otimes \Omega'$ is connected to a product state by a locally generated unitary, i.e.\ it is automorphically equivalent to a product state.  
\item $H$ is frustration-free, i.e. the local terms $h_X$ are all minimized by the state $\Omega$. 
\item The open boundary restrictions (OBC) $H_Z=\sum_{X\subset Z} h_Z$ for balls $Z$ with radius $r$ have a spectral gap $\gamma(r)$ above the ground state sector, that decays no faster than an inverse polynomial in $r$, as $r\to\infty$.
\item The ground states of the open boundary restrictions $H_Z$ for balls $Z$ satisfy the so-called `local topological quantum order' (LTQO) condition: Their local restrictions to a set $X$ approach the local restriction of the state $\Omega$, quasi-exponentially fast, as  $\dist(X,Z^c)\to\infty$. 
\end{enumerate}

The result is that local restrictions of any ground state of $H+J$ approach the local restriction of $\Omega$ fast, as a function of the distance to the impurity or boundary region $\Gamma_j$, see also Figure \ref{fig: main}.

\vspace{3mm}
\textbf{Informal claim on local stability:}
 \emph{There is a stretching exponent $\beta>0$ and finite constant~$C$ such that, for any ground state $\Phi$ of $H+J$, and for any local observable $O_X$ supported in $X\subset \Gamma$ with $\diam(X) \leq \dist(X,\Gamma_j)^{\beta}$, 
$$
|\langle \Phi, O_X \Phi \rangle -  \langle \Omega, O_X \Omega \rangle| \leq  C ||O_X ||  e^{-(\dist(X,\Gamma_j))^{\beta}}
$$}
\vspace{2mm}

\begin{figure}
\centering
\includegraphics[scale=1]{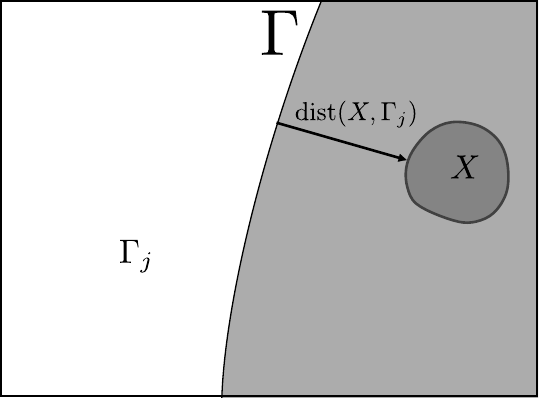}
\caption{The informal claim on local stability concerns ground state expectation values of observables supported in the region $X$, far away from $\Gamma_j$, the support of $J$.}
\label{fig: main}
\end{figure}
As far as we can see, there is no natural sense in which $\Phi$ resembles $\Omega$, other than what is expressed in this claim. In particular, $\Phi$ is not necessarily the unique ground state, it has in general no gap\footnote{Actually, our assumptions do not explicitly require a global gap for $H$ either, but such a gap is there in spirit because of the invertibility assumption,  and we have left it in the abstract to fix thoughts.}, no short range entanglement properties and it is not frustration-free.  To illustrate this, in particular the possibility of long-range entanglement, we consider $\Gamma=\{1,2,\ldots,L\}$ and we imagine that at each site there is a qubit, i.e.\ a two-dimensional Hilbert space, with base states $\ket{\downarrow},\ket{\uparrow}$. The Hamiltonian is 
$H=\sum_{i\in\Gamma} \sigma^z_i$ where $\sigma^z_i$ acts at site $i$, as $\sigma^z_i\ket{\uparrow}_i=\ket{\uparrow}_i$ and $\sigma^z_i\ket{\downarrow}_i=-\ket{\downarrow}_i$.
This Hamiltonian has a unique ground state that is a product, namely $\otimes_i\ket{\downarrow}_i$ .
The perturbation $J$ is chosen as $J= -\sigma^z_1-\sigma^z_L+e^{-cL}j_{1,L}$ with some $c>0$ and $j_{1,L}$ a two-qubit operator acting on sites $1$ and $L$. Because of the very small prefactor $e^{-cL}$, this perturbation $J$ satisfies any reasonable locality requirement.  By taking $j_{1,L}=0$, the perturbed system $H+J$ has 4-fold ground state degenerary.  By taking $j_{1,L}$ to be minus the rank-1 projector on the entangled Bell state $ \tfrac{1}{\sqrt{2}} \left( \ket{\uparrow}_1 \otimes \ket{\uparrow}_L+ \ket{\downarrow}_1 \otimes \ket{\downarrow}_L\right)$, the perturbed system  $H+J$ has a unique ground state that has maximal entanglement between qubits $1$ and $L$.

\subsection{Discussion}
 
 \subsubsection{Spectral stability of the ground state sector}\label{sec: spectral stability}
The above result should be contrasted with recent work on the stability of the spectral gap against locally small perturbations. That recent work can be split in two classes. The first class concerns systems where the unperturbed ground state is a product, see \cite{de2017exponentially,del2020lie,froehlich2020lie,yarotsky2006ground,datta1996low,borgs1996low,lapa2021stability}. The second class is based on the so-called  Bravyi-Hastings-Michalakis (BHM) argument which applies to frustration-free ground states, see  \cite{bravyi2010topological,michalakis2013stability,nachtergaele2020quasi,nachtergaele2021stability}.  Both classes are eventually based on some form of spectral perturbation theory applied to the ground state sector.
In all of the above cases, one also obtains some form of locality for the action of small perturbations, captured by the slogan `local perturbations perturb locally', see also \cite{bachmann2012automorphic,de2015local}. 

  Our result is different in spirit, as our perturbations are not assumed to be small, but we obtain only information about the perturbed ground states in regions far away from the region $\Gamma_j$ where the perturbation acts.  On the other hand, and in contrast to the quoted works, our result relies crucially on the variational principle and so it holds only for ground states (or ceiling states).  One should realize that it is not only the proof of the results in \cite{bravyi2010topological,michalakis2013stability,nachtergaele2020quasi} that does not apply in our case, but also the results should not be expected, as already stressed above through an example.

\subsubsection{Role of the invertibility assumption}
If one drops the invertibility assumption, then the claim on local stability in Section \ref{sec: informal statement on stability} can no longer hold.
To see this, one can consider the two-dimensional toric code model \cite{kitaev2006anyons,kitaev2009topological} on a large square $\Gamma$.  This model is frustration-free and with appropriate boundary conditions, there is a unique groundstate  $\Omega$.  The OBC restrictions always have a spectral gap that is bounded below by a positive constant,  and the topological order condition holds.  However, the state $\Omega$ is not invertible, regardless of boundary conditions.  The lack of invertibility is associated with anyonic excitations that are created in pairs.
By choosing a boundary condition that fixes a single anyon excitation at the boundary of a region $\Gamma$, one forces a partner anyon to be present somewhere in the interior of $\Gamma$, see Figure \ref{fig:Toric}. In this case, there will be multiple ground states, labelled by the position $x$ of the partner anyon. Since $x$ can be arbitrarily far from the boundary, the claim on local stability does not hold. 
\begin{figure}
\centering
\includegraphics[scale=1]{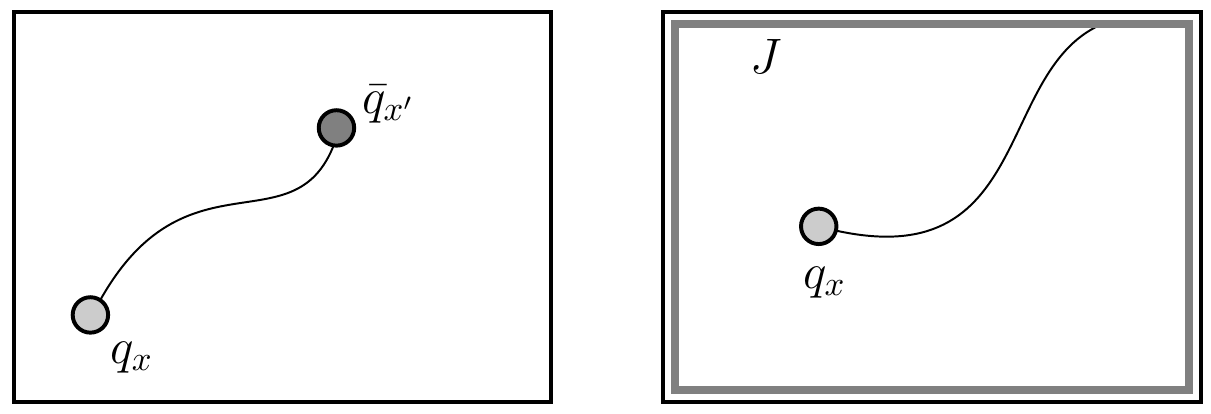}
\caption{Left: a pair of anyons at sites $x,x'$ connected by a (locally non-observable) string. This is an excited state of the toric code Hamiltonian. Right: one of the anyons is fixed at the boundary by the boundary condition. This constrains the partner anyon to be present in the ground state, at some arbitrary site $x$.}
\label{fig:Toric}
\end{figure}
Of course, the large degeneracy of ground states is not a robust feature of the model. Indeed, if one adds a small global perturbation to the model, which will in general give the hitherto static anyons some mobility (see e.g.\ \cite{nachtergaele2020dispersive}), then the degeneracy will be lifted and the unique ground state will correspond to a delocalized anyon.  In this case, the claim on local stability is still violated because the ground state now differs locally from $\Omega $ by an amount that is an inverse polynomial in the volume $|\Gamma|$.

\subsubsection{Remaining assumptions, extensions, and related results}
As argued above, the invertibility assumption can not be dropped. We can also not drop the LTQO condition. Consider for example $H=H_0+B$ acting on the chain of qubits (as the example immediately below the informal claim). We take  $H_0$ the classical Ising Hamiltonian $H_0=-\sum_{i=1}^{L-1}\sigma_i^z \sigma_{i+1}^z$ such that its ground states exhibits spontaneous symmetry breaking, with ground states $\otimes_i\ket{\downarrow}_i$ and $\otimes_i\ket{\uparrow}_i$. The term $B=\sigma^z_1$ breaks the symmetry at the boundary, so that $H$ has a unique ground state $\Omega=\otimes_i\ket{\downarrow}_i$.  In that case our result does not hold, as we see by taking $J=-B$. The LTQO assumption fails to hold since in the interior of $\Gamma$, the symmetry remains unbroken and the OBC restrictions in the interior also have $\otimes_i\ket{\uparrow}_i$ as ground states.  Of course, one could imagine replacing the LTQO condition by some form of translation invariance of the Hamiltonian $H$, ruling out the above example.  

 Apart from this, we do not see that the other assumptions are fundamentally necessary, though they are necessary for our proof.  However, a mild extension of our result seems within reach: we hope that techniques introduced within the BHM argument, see Section \ref{sec: spectral stability}, could allow to generalize our result in the following way: If we were to tighten our assumptions so that the local gap $\gamma(r)\geq \gamma>0$, i.e.\ the local gap is uniformly bounded from below, then the result should also hold for Hamiltonians of the type $H'=H+\epsilon K$ where $H$ and its ground state satisfy the tightened assumptions, $K$ is a Hamiltonian with the same locality restrictions as $H$ and $J$, and $\epsilon \ll 1$ is a small parameter.  
 
Results similar to ours have recently been obtained independently in \cite{henheik2021local}, as a corollary to  \cite{yarotsky2006ground}.
The result of \cite{henheik2021local,yarotsky2006ground} yields full exponential decay and requires no explicit frustration freeness, but it is restricted to weakly interacting spins, i.e.\ perturbations of products.  Instead, our result relies on the automorphic equivalence 
of the unperturbed ground state to a product state, possibly upon adjoining an auxiliary state. Both results need hence an underlying product structure.   Another line of research that is loosely connected to ours, concerns stability at nonzero temperature in spatial dimension 1, see \cite{hastings2007quantum,kato2019quantum}.

\subsection*{Acknowledgements}
W.D.R thanks Marius Sch{\"u}tz and  Simone Warzel for collaboration and discussions of many years ago that proved crucial for the present paper.  We thank two anonymous referees for very constructive remarks that allowed us to substantially improve the paper, in particular concerning the necessity of the assumptions. 
 M.F. was supported in part by the NSF under grant DMS-1907435. W.D.R. was supported in part by the FWO under grant G098919N. S.B. was upported by NSERC of Canada.

\section{Setup}
\subsection{Preliminaries}\label{sec: prelim}
\subsubsection{Spatial structure}\label{sec: spatial structure}
We consider a finite graph $\Gamma$, equipped with the graph distance $\dist(\cdot,\cdot)$. Let
 $B_r(x)=\{y:  \dist(x,y) \leq r \}$ be the ball of radius $r$ centered at $x\in\Gamma$. The graph is assumed to have a finite dimension $d <\infty$, i.e.\ there is a $C_\Gamma<\infty$ such that
$$
\sup_{x\in \Gamma}|B_r(x)| \leq  1+C_{\Gamma} r^{d}.
$$ 
For reasons of recognizability, we refer to vertices as `sites'. 
To every site $x$  is associated a finite-dimensional Hilbert space $\caH_x$ and we define the total Hilbert space $\caH=\caH_{\Gamma}=\otimes_{x\in\Gamma} \caH_x$ and the algebra $\mathcal{A}=\mathcal{A}_\Gamma= \mathcal B(\caH)$ of bounded operators on $\mathcal{H}$. For any $X\subset\Gamma$, we have the Hilbert space $\caH_X=\otimes_{x\in X} \caH_x$ and the algebra $\caA_X =\caB(\caH_X)$.
We use the usual embedding $\caA_X \to \caA$ given by $O_X \to O_X\otimes \ident_{X^c}$.  As is customary, we will identify $O_X\otimes \ident_{X^c}$ with $O_X$ and say that $O_X$ is supported in $X$. 

\subsubsection{Locality}\label{sec: locality}
We write $\bbN^+$ for the strictly positive naturals and we let $\caM$ be a class of functions $m: \bbN^+ \to \bbR^+$ of quasi-exponential decay. The class $\caM$ is defined by the following two conditions.
\begin{enumerate}
\item $m$ is non-increasing. 
\item For every $0< \alpha < 1 $, there exists $C_\alpha, c_\alpha >0$ such that $m(r) \leq C_\alpha e^{-c_\alpha r^\alpha}$. 
\end{enumerate}
Note that the same class of functions is obtained by setting $c_\alpha=1$, which is used often in the proofs. 
To express the locality properties of Hamiltonians, we consider collections $q$ of local terms $(q_X)_{X \subset \Gamma}, q_X \in \caA_X$, sometimes called 'interactions', and we endow them with a family of norms, parametrized by $m \in \caM$:
$$
||q ||_{m}:= \sup_{x\in\Gamma} \sum_{X\ni x} \frac{ ||q_X||}{m(1+\diam(X))}
$$
The locality property is then expressed by the finiteness of $||q ||_{m}$ for some $m\in\caM$.
\subsubsection{Trace norms}\label{sec: trace norms}
We denote by $\tr^{(X)}$ the trace on the Hilbert space $\caH_X$, and we abbreviate $\tr=\tr^{(\Gamma)}$, i.e.\ the trace  on the global Hilbert space $\caH$.
We recall the trace norm $||\cdot ||_{1,X}$ on $\caA_X=\caB(\caH_X)$, defined by 
$$
|| O ||_{1,X} =\tr^{(X)} \sqrt{OO^*},\qquad   O \in \caA_X 
$$ 
We then denote by $\tr_X$ the partial trace $\tr_X: \caA \mapsto \caA_{X_c}$, satisfying $\tr^{(X)} \tr_{X^c} O=\tr O$ for any $O \in \caA$.
Since we will often use trace norms of operators resulting from the partial trace we introduce a short-hand notation:
$$
|O|_X : =   || \tr_{X^c} O ||_{1,X}.
$$
Nonnegative operators $\rho$ on $\caH$ that have unit trace $\tr \rho=1$, are called density matrices and we write $\langle O\rangle_\rho=\tr( \rho O)$ for $O\in\caA$.  For $\Psi\in \caH$, we write $\rho_\Psi$ for the pure density matrix equal to the orthogonal projector on  $\bbC\Psi$. 
\subsection{Spaces and Hamiltonian}

The Hamiltonian $H$ and the perturbation $J$ are written as
\begin{equation}
H=\sum_{X \subset \Gamma} h_X,  \qquad  J=\sum_{X \subset \Gamma} j_X  \qquad h_X=h^*_X \in \caA_X,\quad  j_X=j^*_X \in \caA_X
\end{equation}
The perturbation $J$ is spatially restricted to a region $\Gamma_j \subset \Gamma$ in a mild sense
$$
j_X=0  \quad \text{unless}  \quad   X \cap \Gamma_j \neq \emptyset
$$
Then, we assume, for the collections $h,j$, 
$$
||h||_{m_h},  ||j||_{m_j}  <\infty
$$
for some $m_h,m_j \in \caM$.

\subsection{Assumptions}

Our first assumption states that the Hamiltonian $H$ has a frustration free ground state $\Omega$, i.e.\ $\Omega$ is an eigenvector of $H$ with eigenvalue equal to $\min(\mathrm{spec}(H))$  ($\spec(\cdot)$ is the spectrum). Our conventions are such that this eigenvalue is $0$. 
\begin{assumption}[Frustration-free ground state]\label{ass: ff}
All the terms $h_X$ are nonnegative: $h_X \geq 0$.  
There is a $\Omega \in \caH, ||\Omega||=1$ such that 
  $\Omega \in \Ker(h_X)$ for any $X$. We denote by $\mu=\rho_\Omega$ the corresponding density matrix.
\end{assumption}
We note that the frustration free property depends on the way the Hamiltonian is written as a sum of local terms, i.e.\ on the interaction.
Next, we introduce `open boundary condition' (OBC) restrictions of $H$, 
$$
H_Z=\mathop{\sum}\limits_{X\subset Z} h_X.
$$
Just as the frustration free property, the notion of OBC restriction depends on the interaction. 
We let $P_Z$ be the orthogonal projector on $\Ker (H_Z)$. Since  $h_X\geq 0$, it holds that  
$$\Ker (H_Z)=\mathop{\cap}\limits_{X\subset Z} \Ker(h_X)$$ and hence also
$$
P_Z P_{Z'} =  P_{Z'},\qquad   Z \subset Z'.
$$
The following assumption  has come to be known as `Local Topological Quantum Order' (LTQO) but, in our case, it is better described as the property that the density matrix of any ground state of the OBC Hamiltonian $H_Z$ looks similar to the global ground state $\mu$ in the deep interior of $Z$.  Recall that $B_r(x) \subset \Gamma$ is the ball or radius $r$ centered at $x$.
\begin{assumption}[OBC-regularity]\label{ass: LTO}
There is $m_{O} \in \caM$ and $d_O <\infty$ such that, for any $x\in\Gamma$ and $\Psi \in \Ker(H_{B_r(x)})$, 
$$
| \rho_\Psi-\mu|_{B_{r-k}(x)} \leq  r^{d_O} m_{O}(k),\qquad k<r 
$$ 
\end{assumption}
Let us briefly comment on the precise form of the above bound. We have in mind, roughly, that any $\Psi \in \Ker(H_{B_r(x)})$ differs from $\Omega$ only through the presence of boundary modes. One realization of this would  be that any such $\Psi$ is of the form $e^{iF}\Omega$ with $F$ a sum of terms supported near the boundary of $B_r(x)$. This would indeed lead to the bound above with ${d_O}$ chosen such that $  |\partial B_r(x)| \leq Cr^{d_O}$.

The next assumption concerns the local gap of OBC Hamiltonians. Let 
$$
\gamma(Z)= \min (\spec(H_Z) \setminus \{0\})
$$
be the spectral gap of $H_Z$.  
\begin{assumption}[Local Gap]\label{ass: local gap}
There are $C_\gamma,d_{\gamma} <\infty$ such that, for any $x\in\Gamma$
$$
 \frac{1}{\gamma(B_r(x))}  \leq C_\gamma r^{d_\gamma}.
$$
\end{assumption}
This assumption \ref{ass: local gap} might be misleading. In most examples of models with gapped frustration-free ground states, that we are aware of, OBC restrictions $H_Z=\sum_{X\subset Z}h_Z$ have a gap that is actually bounded below in $Z$, and the physical edge modes are found as ground states of $H_Z$, i.e.\ the eigenvalue at $0$ will typically be degenerate. In that case, the genuine restriction is not so much assumption \ref{ass: local gap} but rather assumption \ref{ass: LTO} which describes the kernel $\Ker(H_{B_r(x)})$.  

\subsubsection{Invertibility}
Let $\caH'$ be $\caH'=\caH'_\Gamma = \otimes_{x\in\Gamma}\caH'_x$ with $ \caH'_x$  finite-dimensional Hilbert spaces.  We will now consider the tensor product $\caH \otimes  \caH'$ which we denote by
$$
\widetilde \caH= \caH \otimes^{\mathrm{aux}}  \caH'.
$$
The superscript $\otimes^{\mathrm{aux}} $ reminds us of the fact that this is not a tensor product between disjoint spatial regions, but between the `original' Hilbert space and an `auxiliary' Hilbert space.
This is helpful because we also view $\widetilde \caH$ again as a tensor product over sites
$$
\widetilde \caH= \otimes_{x \in \Gamma}  \widetilde \caH_x,\qquad   \widetilde \caH_x = \caH_x \otimes\aux \caH'_x
$$

\begin{figure}
\centering
\includegraphics[scale=1]{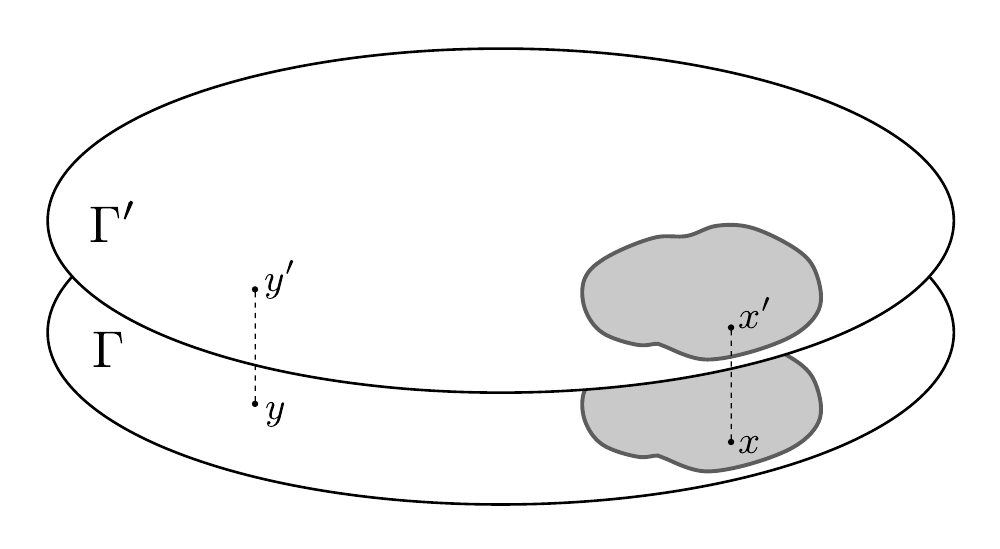}
\caption{The invertibility assumption uses an auxiliary Hilbert space $\caH'$ that has the same spatial structure as the original space $\caH$, i.e.\ $\Gamma'$ is a copy of $\Gamma$.}
\label{fig:invertibility}
\end{figure}

The local structure of the space $\widetilde \caH$ is analogous to the one of $\caH$, except that now the dimension of the on-site spaces is larger, see Figure \ref{fig:invertibility}. We consistently put a prime on algebras derived from $\caH'$, $\caA'=\caB(\caH')$ and we put a tilde on algebras derived from $\widetilde\caH$, i.e.\ $\widetilde \caA =\mathcal{B}(\widetilde \caH)$ and $\caA_X' \subset \caA',\widetilde \caA_X \subset \widetilde \caA$ are subalgebras of operators supported in $X$.  We also copy the definition of collections $q$ of local operators $(q_X)_{X}$ and their norms $||\cdot ||_m$, simply replacing $\caA_X$ by $\widetilde \caA_X$.    The next assumption expresses that the state $\Omega$ has no \emph{long-range entanglement}. Loosely speaking, one could describe this assumption as `upon adjoining an auxiliary state $\Omega'$, $\Omega \otimes \Omega'$ is automorphically equivalent to a product state'  \cite{bachmann2012automorphic,hastings2005quasiadiabatic}.
\begin{assumption}[Invertibility]\label{ass: invertibility}
There are  collections $q(s)$ of local operators $(q_X(s))_{X\subset\Gamma}$, indexed by $s\in[0,1]$,  such that
\begin{enumerate}
\item $q_X(s)=q^*_X(s) \in\widetilde\caA_X$ and $s\mapsto q_X(s)$ is measurable, for any $X$. 
\item $\sup_{s\in[0,1]}  || q(s) ||_{m_I}  \leq C_I <\infty$ for some $m_I\in\caM$.
\item There is a product state $\Pi=\otimes_{x\in \Gamma}\Pi_x \in \widetilde \caH$ and a state $\Omega' \in\caH'$ such that
$$
\Omega \otimes\aux \Omega'=U(1) \Pi,\qquad    U(s)=\ident+i\int_0^s du \,  H_q(u) U(u)
$$ 
where $H_q(s)=\sum_{X \subset \Gamma}  q_X(s)$. 
\end{enumerate}
We say that $\Omega'$ is an `inverse' to $\Omega$. 
\end{assumption}
The `invertibility' assumption roughly corresponds to states that do not have anyonic excitations. Examples of invertible states are: symmetry protected topological states (SPT's) \cite{freed2014anomalies,KitaevConf,gu2009tensor}, states characterized by an integer quantum Hall effect \cite{kapustin2020hall}, etc..   We refer to the extensive literature for a more thorough discussion.  
The motivation for considering the type of spatial decay expressed by the class $\caM$ is because this corresponds to the decay one can prove for the spectral flow connecting ground states of a uniformly gapped Hamiltonian, by the technique of (quasi-)adiabatic continuation, see \cite{hastings2005quasiadiabatic,bachmann2012automorphic}.

\subsection{Result} \label{sec: result}

We say that a unit vector $\Phi \in \caH_\Gamma$ is a ground state of the perturbed Hamiltonian $H+J$ if $\Phi$ is an eigenvector of $H+J$ with eigenvalue equal to $\min{\mathrm{spec}(H+J)}$.
Recall that $\mu$ is the density matrix associated to $\Omega$, the distinghuished ground state of the unperturbed Hamiltonian $H$, and $\rho_\Phi$ is the density matrix associated to $\Phi$. 
By a `constant', we mean a quantity that can depend only on $m_h,m_j,m_I \in \caM$, the numbers $C_\gamma,C_\Gamma$ and  $d,d_O,d_\gamma$. In particular, constants can not depend on the sizes $|\Gamma|$ {and $|\Gamma_j|$}.
\begin{theorem}\label{thm: main}
For any $w>0$, there is a constant $C(w)$ such that, for
any ground state $\Phi$ of $H+J$, and any $x\in \Gamma$, with $R:=\dist(x,\Gamma_j)$,
$$
|\rho_{\Phi}-\mu|_{B_{(R^{p}/2)}(x)}   \leq C(w) e^{-R^{p-w}},\qquad  p= \frac{1}{d+d_\gamma+2}
$$
\end{theorem}
To connect this theorem to the informal claim on stability in Section \ref{sec: informal statement on stability}, note that the left hand side is equal to $\sup_{O \in  \caA_X, ||O||=1}  |\langle \Phi, O \Phi \rangle -  \langle \Omega, O \Omega \rangle|  $ with $X=B_{(R^{p}/2)}(x)$.

\section{Proof of Theorem \ref{thm: main}}

We define some additional notation to be used in this section. 
For regions $Z \subset\Gamma$, we define the $r$ fattening
$$
(Z)_r :=\{ x\in\Gamma, \dist(x,Z) \leq r\}
$$
(we will simply write $(Z)_r=Z_r$ when no confusion is possible)
and the boundary 
$$
\partial Z= (Z)_1 \cap  (Z^c)_1
$$
We use constants $C,c$ as introduced just before Theorem \ref{thm: main}, i.e.\ depending on a number  of fixed parameters. Sometimes we write $C(a)$ to indicate that $C$ can additionally depend on a parameter $a$. We also use the generic notation $m$ for a function in $\caM$ that depends possibly on those same fixed parameters.  Just as the constants $C,c$, the precise function $m$ can change from line to line.  In the same vein, we also use $p$ to denote a polynomial that depends only on the fixed parameters. 
%

\subsection{Stitching maps}\label{sec: stitching maps}

Recall that $\mu$ is the global ground state density matrix of the unperturbed Hamiltonian $H$. 
We say that a family of maps $\Sigma_Z$, indexed by $Z\subset \Gamma$, are stitching maps if the they satisfy the following properties.

\begin{definition} \label{def: locality stitch}
Stitching maps $\Sigma_Z$ with $Z \subset \Gamma$ are trace-preserving completely positive maps $\caA \to \caA $ iff.\  there is an $m\in\caM$ such that, for any $Z,X \subset \Gamma$ and any density matrices $\rho,\omega$ on $\caH$,
\begin{enumerate}
\item  $|\Sigma_Z(\rho)-\tr_{Z^c} \mu \otimes \tr_{Z}\rho|_{X}  \leq |X|  m(\dist(\partial Z,X))$
\item $
|\Sigma_Z (\rho)- \Sigma_Z(\omega)|_{X}  \leq  |\rho- \omega|_{(X)_r} +|X| m(r) $
\item $
\Sigma_Z (\mu)= \mu $
\end{enumerate}
\end{definition}
To state this in a rough way, stitching maps are such that $
\Sigma_Z (\mu)= \mu $ and further
$$
\Sigma_Z(\rho) =\begin{cases}  \mu &   \text{deep inside $Z$} \\      \rho & \text{far outside $Z$} \end{cases}
$$
The stitching maps will be used for regions $Z$ that are far away from the perturbation region $\Gamma_j$.

The purpose of the invertibility assumption \ref{ass: invertibility} is precisely to ensure the existence of stitching maps, as we show now.
Recall that $\widetilde \caH$ is the enlarged Hilbert space.
\begin{proposition} \label{lem: stitching invertible}
For any region $Z\subset\Gamma$, there are unitaries $V=V_Z$ acting on $\widetilde \caH$ such that the following family of CP maps $\Sigma=\Sigma_Z$ are stitching maps in the sense  of definition \ref{def: locality stitch}:
\begin{equation}
\Sigma(\rho)= \tr_{\mathcal{H}'}\left[V^*\left(\tr_{Z^c} (V\mu\otimes\aux\mu'V^*)\otimes \tr_{Z} (V\rho\otimes\aux\mu'V^*)\right)V \right]
\end{equation}
where  $\mu=\rho_{\Omega}$ and  $\mu'=\rho_{\Omega'}$ are density matrices on $\caA$ and $\caA'$ and $\Omega' \in \caH'$ is an `inverse state' to $\Omega$, see assumption \ref{ass: invertibility}.
\end{proposition}
We first construct the appropriate unitaries $V$ featuring in Proposition \ref{lem: stitching invertible}. To that end, we introduce  the truncated interactions $\hat q(s)$ by
$$
\hat{q}_X(s) = \begin{cases}  q_X(s) &  X \subset Z\quad \text{or} \quad X \subset Z^c  \\
0  & \text{otherwise} \end{cases}  
$$
with $q(s)$ as given in assumption \ref{ass: invertibility}.
We let $\hat U(s)$ be the unitary defined analogously to  $U(s)$ in the invertibility assumption \ref{ass: invertibility}, but with $H_{q}(s)$ replaced by $H_{\hat q}(s)= \sum_{X} \hat{q}_X(s)$. We then set 
$$
V(s)= \hat U(s) U^*(s),\qquad    V=V(1)
$$
and we observe that
\begin{equation}\label{eq: spatial split}
V \Omega\otimes^{\mathrm{aux}} \Omega' =  \Omega_Z \otimes \Omega_{Z^c}
\end{equation}
for some $\Omega_Z \in \widetilde\caH_Z, \Omega_{Z^c} \in \widetilde\caH_{Z^c}$. This shows that $\Sigma_Z$ satisfies property 3) of definition \ref{def: locality stitch}.
The full proof of Proposition \ref{lem: stitching invertible}, i.e.\ the verifications of properties 1) and 2) of definition \ref{def: locality stitch},  uses some standard terminology and locality estimates that are not needed in the rest of our proof, hence we postpone them to the appendices, but we do state the relevant locality property of the unitary operators $V$. To that end, let us introduce the conditional expectation, for any $Z\subset\Gamma$,  
$$
\bbE_{Z^c}:\widetilde\caA\mapsto \widetilde\caA_Z:  O \mapsto  \tfrac{1}{N}\tr_{Z^c}(O), \qquad N= \dim (\widetilde \caH_{Z^c}). 
$$
\begin{lemma}\label{lem:splitting2}
For any regions $X,Z$ and $O \in \widetilde\caA_X$, with $V=V_Z$ as defined above, 
\begin{equation}\label{eq: lr}
|| V^*OV -\bbE_{(X_r)^c}(V^*OV) || \leq ||O|| |X| m(r)
\end{equation}
and
\begin{equation}\label{eq: locality of vz}
||V^*OV-O || \leq  ||O || |X|  m(\dist(X,\partial Z)) 
\end{equation}
The same estimates hold\emph{•} as well if we exchange $V$ and $V^*$.
\end{lemma}
The proof of Lemma \ref{lem:splitting2} follows from well-known considerations based on Lieb-Robinson bounds, and it is sketched in Appendix~\ref{sec: locality estimates}.
Finally, the proof of properties 1) and 2) of definition \ref{def: locality stitch} is an intuitive consequence of Lemma \ref{lem:splitting2}. This proof is given in Appendix \ref{sec: locality of cp}.

\subsubsection{Action of a stitching map $\Sigma_Z$ at the stitch $\partial Z$}

We will now establish a property of stitching maps, which we will henceforth denote by $\Sigma_Z$, that relies crucially on the OBC-regularity assumption \ref{ass: LTO}. It makes explicit that stitching maps are `seamless', i.e.\ they do not introduce errors at the cut $\partial Z$.  Recall that $P_X$ is the ground state projector corresponding to the OBC restriction $H_X$ and let $\bar P_X=\ident-P_X$.  

\begin{proposition}\label{lem: stitching}
Let $\sigma=\rho_\Psi$ be the pure density matrix associated to some $\Psi \in\caH$.  Consider a set $Z$ and  balls $B_r=B_r(x)$ for some site $x$.    Then (recall that $p(\cdot)$ is a polynomial),
$$\langle \bar P_{B_{r-k}} \rangle_{\Sigma_Z(\sigma)} \leq   3\langle \bar P_{B_r} \rangle_\sigma+p(r)m(k).
 $$
\end{proposition}
\begin{proof}
Since $Z$ will be fixed, we drop it and write $\Sigma=\Sigma_Z$.
Let us denote $\epsilon:=\langle \bar P_{B_r} \rangle_\sigma $ and remark that $||P_{B_r}\Psi-\Psi|| \leq \sqrt{\epsilon}$.  If $\epsilon=1$, then the required bound is trivial, hence we assume $\epsilon<1$. 
We split 
$$
\sigma= \sum_{i=1}^4  \zeta_i =  P_{B_r} \sigma P_{B_r} +  P_{B_r} \sigma \bar P_{B_r} + \bar P_{B_r} \sigma P_{B_r}+ \bar P_{B_r} \sigma \bar P_{B_r}.
$$ 
We will treat these terms separately, i,e.\ find bounds on
$$
\tr  \bar P_{B_{r-k}} \Sigma(\zeta_i),\qquad i=1,2,3,4.
$$
\textbf{The term $\zeta_1=P_{B_r}\sigma P_{B_r}$}\\
Since $\epsilon<1$,
 $\tr\zeta_1 >0 $ and we can define 
the pure density matrix $\zeta=\zeta_1/\tr\zeta_1$, satisfying $\langle \bar P_B \rangle_\zeta=0$. 
By the properties of definition \ref{def: locality stitch}, we get the equality and first inequality in 
$$
|\Sigma(\zeta)-\mu|_{B_{r-k}}  =     |\Sigma(\zeta)-\Sigma(\mu)|_{B_{r-k}}   \leq   
   |\zeta-\mu|_{B_{r-k/2}} +  p(r) m(k) \leq  p(r) m(k)
$$
whereas the last inequality follows from the OBC regularity assumption \ref{ass: LTO}.   
By the frustration-free property, $ \tr\bar P_{B_{r-k}}\mu=0 $, and thus the above inequality yields $\tr  \bar P_{B_{r-k}} \Sigma(\zeta)  \leq p(r) m(k) $. Since $\zeta_1\leq \zeta$, this implies 
$$
\tr  \bar P_{B_{r-k}} \Sigma(\zeta_1)  \leq p(r) m(k) 
$$
\noindent 
\textbf{The term $\zeta_4= \bar P_{B_r} \sigma \bar P_{B_r}$}\\    Here we have
$\tr\zeta_4 = \tr \bar P_{B_r} \sigma =  \epsilon$ and since $\Sigma$ is trace-preserving,  
$$
\tr  \bar P_{B_{r-k}} \Sigma(\zeta_4)  \leq  \epsilon. 
$$
\noindent
\textbf{The term $\zeta_2= P_{B_r} \sigma \bar P_{B_r}$}\\ Since $\sigma$ is pure, we can write
$$
\zeta_2= P_{B_r} \sigma \bar P_{B_r} =   \langle \bar P_{B_r} \rangle_\sigma^{1/2} \,     |a \rangle\langle b|,  \qquad ||a|| \leq ||b||=1, \qquad   P_{B_r} a=a
$$
and we note that $|a\rangle\langle a|=\zeta_1$.  To estimate $\tr \bar P_{B_{r-k}}\Sigma{ (|a \rangle\langle b|)}$, we use Lemma \ref{lem: stinespring} below to get
$$
|\tr \bar P_{B_{r-k}} \Sigma(\zeta_2 )| \leq  \sqrt{ \langle \bar P_B\rangle_\sigma }  \sqrt{\tr (\bar P_{B_{r-k}}\Sigma(\zeta_1))}
\leq  \sqrt{\epsilon} \sqrt{p(r)m(k)}   \leq \epsilon+p(r)m(k) 
$$
The second inequality follows from the bound for the term $\zeta_1$ above.
\begin{lemma}\label{lem: stinespring}
For an orthogonal projector $P$,  a trace-conserving CP map $\Gamma$ and $||a||,||b|| \leq 1$, we have
$$
|\tr (P\Gamma(|a \rangle\langle b|))| \leq  \sqrt{\tr (P\Gamma(|a \rangle\langle a|))   \tr (P\Gamma(|b \rangle\langle b|))} \leq  \sqrt{\tr (P\Gamma(|a \rangle\langle a|)) }
$$
\end{lemma}
\begin{proof}
Introducing the adjoint map $\Gamma^*$, the first inequality reads 
$$ |\langle b, \Gamma^*(P)a \rangle| \leq  \sqrt{\langle a, \Gamma^*(P)a \rangle \langle b, \Gamma^*(P)b \rangle }.$$ It follows from  positivity of $\Gamma^*(P)$ and the Cauchy-Schwarz inequality.  
\end{proof}
\noindent
\textbf{The term $\zeta_3= \bar P_{B_r} \sigma P_{B_r}$}\\
The bound and its proof are  analogous to the case of $\zeta_2$. 
The claim of the lemma follows by adding the bounds on  $\tr \bar P_{B_{r-k}} \Sigma(\zeta_i )$ for $i=1,2,3,4$.
 \end{proof}

\subsection{A bound on the energy increase due to stitching}\label{sec: basic inequality}
If we apply the stitching map $\Sigma_Z$ to a pure density matrix $\sigma$, we expect its energy (associated to $H$) to decrease in the far interior of $Z$ and to remain unchanged far outside $Z$. Around the stitch $\partial Z$, the energy can of course increase. In the present section we prove that this energy increase is upper bounded by the energy of $\sigma$ around the stitch. We make this precise in Lemma \ref{lem: error stitch}. The proof is based largely on Proposition \ref{lem: stitching} that allows us to convert a bound on the energy to a bound on the local deviation from $\mu$.

\begin{figure}
\centering
\includegraphics[scale=1.3]{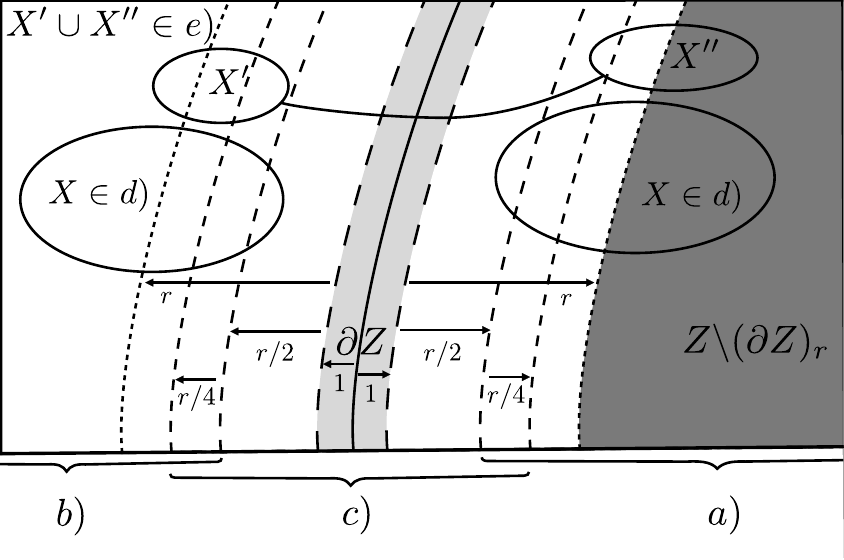}
\caption{The classes of sets $X$ in the proof of Lemma \ref{lem: error stitch}.}
\label{fig:lemma_error_stitch}
\end{figure}
  
\begin{lemma} \label{lem: error stitch}
Let $\sigma$ be a pure density matrix, i.e.\ $\sigma=\rho_\Psi$ for some $\Psi \in\caH$, and let $\Sigma=\Sigma_Z$ for some region $Z \subset \Gamma$.  Then 
$$
\sum_{X \not\subset Z\setminus (\partial Z)_r}   \left(\langle h_X \rangle_{\Sigma(\sigma)}- \langle h_X \rangle_{\sigma}\right)
 \leq Cr^{d+d_{\gamma}}\langle H_{ (\partial Z)_{r}} \rangle_\sigma +  m(r) |\partial Z|
$$
\end{lemma}

\begin{proof}
We split the class of  $X$ contributing the left hand side\footnote{The condition $X \not\subset Z\setminus (\partial Z)_r$ appearing in the statement of the lemma, is actually irrelevant for its validity.} in disjoint classes $a,b,c,d,e$: 
\begin{itemize}
\item[$a)$] $X \subset Z\setminus (\partial Z)_{r/2}$
\item[$b)$]   $X \subset Z^c\setminus (\partial Z)_{r/2}$
\item[$c)$]  $ X \cap (\partial Z)_{r/2} \neq \emptyset$ and $\diam(X)< r/4$. 
\item[$d)$]   $X \cap (\partial Z)_{r/2} \neq \emptyset$ and $\diam(X)\geq r/4$.
\item[$e)$]   $X \subset ((\partial Z)_{r/2})^c $ and $X \cap Z \neq \emptyset$ and $X \cap Z^c \neq \emptyset$ .  
\end{itemize}
 
We will now estimate, with $x=a,b,c,d,e$,
\begin{equation}\label{eq: different classes}
\sum_{ X \in \text{class} \, x}   \left(\langle h_X \rangle_{\Sigma(\sigma)}- \langle h_X \rangle_{\sigma}\right)
\end{equation}

\begin{enumerate}
\item[\textbf{class $a$}]
\begin{align*}
\sum_{ X \in \text{class} \, a}   \left(\langle h_X \rangle_{\Sigma(\sigma)}- \langle h_X \rangle_{\sigma}\right)
& \leq 
\sum_{X \subset Z\setminus (\partial Z)_{r/2}}  \langle h_X \rangle_{\Sigma(\sigma)} \\
& \leq 
\sum_{X \subset Z\setminus (\partial Z)_{r/2}}  \left(\langle h_X \rangle_{\mu} + ||h_X|| |X| m(\dist(X,\partial Z)\right)  \\
 & \leq   \sum_{x: \dist(x,\partial Z) \geq r/2}  m(\dist(x,\partial Z)
\sum_{X \ni x}   ||h_X|||X| \leq
  |\partial Z|m(r)
\end{align*}
The first inequality follows from non-negativity of $h_X$. The second from  
property 1) of definition \ref{def: locality stitch} and the third from $\langle h_X\rangle_\mu=0$.  Then we use $ \sup_{x} \sum_{X \ni x}   ||h_X|||X| \leq C$ and the rapid decay of $m$ to absorb a polynomial factor in $r$.  
\item[\textbf{class $b$}]
\begin{align*}
\sum_{ X \in \text{class} \, b}   \left(\langle h_X \rangle_{\Sigma(\sigma)}- \langle h_X \rangle_{\sigma}\right)
&\leq 
 \sum_{X \subset Z^c\setminus (\partial Z)_{r/2}}  ||h_X|| |X|  m(\dist(X,Z)) \\
  &\leq  \sum_{x:\dist(x, Z)\geq r/2}  m(\dist(x,Z))  \sum_{ X \ni x }  ||h_X|| |X|  \leq |\partial Z|  m(r) 
\end{align*}
The first inequality is property 1) of definition \ref{def: locality stitch}, the rest is analogous to case $a$ above. 
\item[\textbf{class $c$}]
\begin{align*}
\sum_{ X \in \text{class} \, c}   \left(\langle h_X \rangle_{\Sigma(\sigma)}- \langle h_X \rangle_{\sigma}\right)
 & \leq \sum_{x\in (\partial Z)_{r/2}} \quad \sum_{ X\ni x,\diam (X) \leq r/4}    \langle h_X \rangle_{\Sigma(\sigma)}  \\
 &   \leq \sum_{x\in (\partial Z)_{r/2}}  \quad \sum_{X\ni x,\diam (X) \leq r/4}  ||h_X||  \langle \bar P_{X} \rangle_{\Sigma(\sigma)}    \\
  & \leq \sum_{x\in (\partial Z)_{r/2}}  C \langle \bar P_{B_{r/4}(x)} \rangle_{\Sigma(\sigma)} \\
  & \leq \sum_{x\in (\partial Z)_{r/2}}   C   \langle \bar P_{B_{r/2}(x)} \rangle_{\sigma} +m(r)   \\
   & \leq  \sum_{x\in (\partial Z)_{r/2}}   Cr^{d_\gamma} \langle  H_{B_{r/2}(x)} \rangle_{\sigma} +m(r)   \\ 
    & \leq  Cr^{d+d_\gamma} \langle  H_{(\partial Z)_r} \rangle_{\sigma} +   |(\partial Z)_{r/2}| m(r)   \\   
    & \leq C r^{d+d_\gamma} \langle  H_{(\partial Z)_r} \rangle_{\sigma} +   |\partial Z| m(r)   \\    
\end{align*}
The first and second inequality is by nonnegativity of $h_X$. The third inequality uses $\bar P_X \leq \bar P_{X'}$ for $X\subset X'$, i.e.\ the frustration freeness of assumption \ref{ass: ff}, and  $\sum_{X\ni x} ||h_X||<C$. The fourth inequality uses Proposition \ref{lem: stitching}. The fifth inequality is by the local gap assumption \ref{ass: local gap}.  The sixth inequality follows because any $X\subset (\partial Z)_r$ is included in at most a number $ Cr^d$ of  balls with radius $r/2$. 
For the last inequality, we use $|(\partial Z)_{r/2}| \leq p(r)|\partial Z|$  and we absorbed the polynomial in $m$.

\item[\textbf{class $d$}]
\begin{align*}
\sum_{ X \in \text{class} \, d}   \left(\langle h_X \rangle_{\Sigma(\sigma)}- \langle h_X \rangle_{\sigma}\right)
&\leq  \sum_{x \in (\partial Z)_{r/2}}  \sum_{X \ni x: \diam(X)\geq r/4 }   2 || h_X|| 
 \\
 &  \leq  \sum_{x \in (\partial Z)_{r/2}} m(r)  \leq  |\partial Z| m(r)
\end{align*}
The first inequality follows from $\langle h_X \rangle_{\Sigma(\sigma)}, \langle h_X \rangle_{\sigma}\leq ||h_X||$.  Then we use analogous reasoning as  for \textbf{class $a$}.
\item[\textbf{class $e$}]
\begin{align*}
\sum_{ X \in \text{class} \, e}   \left(\langle h_X \rangle_{\Sigma(\sigma)}- \langle h_X \rangle_{\sigma}\right)
\leq   \quad  \sum_{x: \dist(x,\partial Z) \geq r/2} \quad 
\sum_{X \ni x:  \diam(X) \geq  \dist(x,\partial Z)  }   2 || h_X||  \leq |\partial Z| m(r)
\end{align*}
The reasoning is analogous to that  for \textbf{class $d$} and \textbf{class $a$}.
\end{enumerate}

\end{proof}
The previous Lemma \ref{lem: error stitch} was only concerned with the unperturbed Hamiltonian $H$. We now state an analogous estimate for the Hamiltonian $J$.
\begin{figure}
\centering
\includegraphics[scale=0.9]{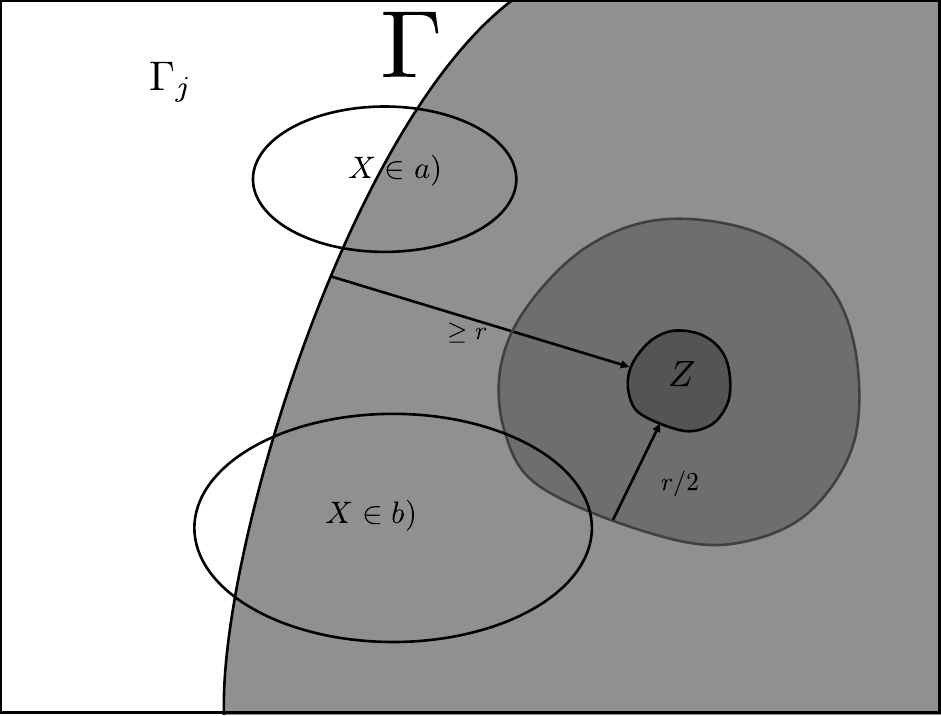}
\caption{The classes of sets $X$ in the proof of Lemma \ref{lem: error stitch j}.  We only need to consider $X$ such that $X\cap \Gamma_j\neq\emptyset$, by the definition of $j$.}
\label{fig:lemma_error_stitch_j}
\end{figure}

\begin{lemma} \label{lem: error stitch j}
Let $\sigma$ be a density matrix and let $Z$ be a set such that $\dist(Z,\Gamma_j) \geq r$.
Then
\begin{equation}\label{eq: drop j}
|\langle J \rangle_{\Sigma_Z(\sigma)}- \langle J \rangle_{\sigma}| \leq  m(r) |\partial Z|
\end{equation}
\end{lemma}

\begin{proof}
We estimate 
$$|\langle J \rangle_{\Sigma_Z(\sigma)}- \langle J \rangle_{\sigma}|  \leq \sum_{X} |\langle j_X \rangle_{\Sigma_Z(\sigma)}- \langle j_X \rangle_{\sigma}|$$.  We split the sets $X$ into two disjoint classes: 
\begin{itemize}
\item[$a)$]  $\dist(X,Z) \geq r/2$
\item[$b)$]  $\dist(X,Z) < r/2$
\end{itemize}
We estimate the contribution of each class and we get in each case the desired bound.
\begin{enumerate}
\item[\textbf{class $a$}]
\begin{align*}
\sum_{X \in \text{class a}} |\langle j_X \rangle_{\Sigma_Z(\sigma)}- \langle j_X \rangle_{\sigma}| 
& \leq  \sum_{x: \dist(x,Z)>r/2 }  m(\dist(x,Z))  \sum_{X \ni x}    || j_X ||   \\
& \leq  C \sum_{x: \dist(x,Z)>r/2 }  m(\dist(x,Z))   \leq m(r) |\partial Z|  \\
\end{align*}
We use property 1 of definition \ref{def: locality stitch} to get the first inequality. The last inequality follows by the rapid decay of $m$ and the fact that $\Gamma$ has finite dimension $d$. 
%
%
\item[\textbf{class $b$}]
\begin{align*}
\sum_{X \in \text{class b}} |\langle j_X \rangle_{\Sigma_Z(\sigma)}- \langle j_X \rangle_{\sigma}|
& \leq  \sum_{k\geq r/2} \, \sum_{x: \dist(x,\Gamma_j)=k  } \,  \sum_{X\ni x: \diam(X) \geq k} 2 ||j_X||  \\
& \leq \sum_{k\geq r/2} \, \sum_{x: \dist(x,\Gamma_j)=k  }  m(k)   \leq   \sum_{k\geq r/2} Ck^d |\partial Z| m(k)  \leq m(r)|\partial Z| \\
\end{align*}
The first inequality follows because $j_X=0$ unless $X \cap \Gamma_j \neq \emptyset$.  In the last inequality, we again absorbed a polynomial factor in $m$.
\end{enumerate}
  
\end{proof}

\subsection{Isoperimetry}
From now on, we assume that $\sigma=\rho_{\Phi}$ with $\Phi$ a ground state of $H+J$. In particular, this means that the following inequality holds: 
\begin{equation}
\langle H+J \rangle_{\Sigma_Z(\sigma)}- \langle H+J \rangle_{\sigma}\geq 0.
\end{equation}
We will use Lemma's \ref{lem: error stitch} and \ref{lem: error stitch j} to relate the energy in a region $Z$ to the energy around $\partial Z$, which explains the title of this subsection.  We again choose the region $Z$ to be far away from $\Gamma_j$: $\dist(Z,\Gamma_j)\geq r$ with $r \gg 1$. 
By Lemma \ref{lem: error stitch j}, we can discard $J$ in the above inequality at small cost, and we obtain 
\begin{equation}
\langle H \rangle_{\Sigma_Z(\sigma)}- \langle H \rangle_{\sigma} + m(r) |\partial Z| \geq 0
\end{equation}
We split this inequality as
\begin{equation}\label{eq: inequality with error}
\langle H_{Z\setminus (\partial Z)_r} \rangle_{\Sigma_Z(\sigma)}- \langle H_{Z\setminus (\partial Z)_r} \rangle_{\sigma} +   \sum_{X \not\subset Z\setminus (\partial Z)_r}   \left(\langle h_X \rangle_{\Sigma_Z(\sigma)}- \langle h_X \rangle_{\sigma}\right) + m(r)|\partial Z|  \geq 0
\end{equation}
from which we get
%
%
%
\begin{lemma}  \label{lem: inequality iterated}
$$
\langle H_{Z\setminus (\partial Z)_r}\rangle_\sigma \leq  Cr^{d+d_{\gamma}} \langle H_{ (\partial Z)_{r}} \rangle_\sigma +  m(r) |\partial Z|
$$
\end{lemma} 
\begin{proof}
We start from inequality \eqref{eq: inequality with error}. 
 By the same easy reasoning as used in the proof of lemma \ref{lem: error stitch} for \textbf{class $a$}, we bound the term $\langle H_{Z\setminus (\partial Z)_r} \rangle_{ \Sigma(\sigma)}$ by $m(r) |\partial Z|$. Then the upper bound of Lemma \ref{lem: error stitch} gives us immediately the desired claim. 
\end{proof}

Let us now put Lemma \ref{lem: inequality iterated} to use. 
We will consider a sequence of regions $Z_i$ chosen as concentric balls $Z_i=B_{R_i+r}=B_{R_i+r}(x)$, with $x$ fixed and $R_i$ to be specified.
We note that 
$$
B_{R_i} = Z_i \setminus (\partial Z_i)_{r-1}
$$
and we abbreviate
$$
E_i= \langle H_{B_{R_i}} \rangle_\sigma, \qquad \delta_i=  \langle H_{(\partial B_{R_{i}+r})_{r-1  }} \rangle_\sigma  
$$
Then, the inequality in Lemma \ref{lem: inequality iterated} reads
\begin{equation}\label{eq: inequality reloaded}
E_{i}\leq Cr^{d+d_\gamma} \delta_i +  m(r)|\partial B_{R_i+r}|   
\end{equation}
Note also that  $E_{i+1} \geq E_i+ \delta_i$ provided that $R_{i+1}\geq R_{i}+2r$. Our strategy will be to establish a lower bound on $E_i$, depending on $E_1$, and then eventually use the a priori upper bound
$$
E_i =\langle H_{B_{R_i}} \rangle_\sigma \leq ||H_{B_{R_i}}  || \leq  CR_i^d
$$
to get an upper bound on $E_1$. 
Proceeding in this way, we obtain   
\begin{lemma}\label{lem: isoperimetry iterated}
Let $R=\dist(x,\Gamma_j)$ with $x$ the center of the balls above. Then, for any $w>0$
$$
\langle H_{B_{R^p}} \rangle_\sigma   \leq C(w) e^{-R^{p-w}},\qquad  p= \frac{1}{d+d_\gamma+2}.
$$
\end{lemma}
\begin{proof}

We choose $R_i=(2i-1)r$ for $i=1,\ldots, i_*$ with  $i_*$ the largest integer that is smaller than $R/(2r)$. Using inequality \eqref{eq: inequality reloaded} and $E_{i+1} \geq E_i+ \delta_i$ we obtain
\begin{equation} \label{eq: start gronwall}
E_{i+1} \geq a E_i - b_i,\qquad   a = (1+ c r^{-d - d_\gamma}),\qquad  b_i = m(r) |\partial B_{R_i+r}| 
\end{equation}
where we also updated the function $m$ compared to \eqref{eq: inequality reloaded}. 
We will now use a Gr{\"o}nwall-type inequality. 
Multiplying the  inequality \eqref{eq: start gronwall} by $a^{-i}$ and summing over $i$ yields
$$\sum_{i=1}^k E_{i+1}a^{-i} \geq \sum_{i=1}^k a^{-i+1} E_i + \sum_{i=1}^k b_i a^{-i},$$
which implies 
$$
E_i a^{-i+1} \geq  E_1 - \sum_{j=1}^{i-1} b_j a^{-j}.
$$
and hence, choosing $i=i_*$,
$$
E_1 \leq E_{i_*} a^{-{i_*}+1} + \sum_{j=1}^{{i_*}-1} b_j a^{-j}.
$$
To estimate the sum on the right hand side, we bound $|\partial B_{R_j+r}| \leq C (3jr)^d$ and estimate
$$
\sum_{j=1}^{{i_*}-1} b_j a^{-j} \leq m(r) r^d \sum_{j=1}^{\infty} j^d a^{-j} \leq  m(r).
$$
 Using in addition the a priori bound $E_{i_*}\leq \langle H_{B_R}\rangle_\sigma < CR^d$, we get
$$
E_{1} \leq C R^d a^{1-\frac{R}{2r}} + m(r) \leq CR^d e^{-c R r^{-(d + d_\gamma+1)}} +  m(r).
$$
To optimize the inequality we let $r$ grow with $R$, namely $r=R^{p}$. For $m\in \caM$ and $w>0$, we can find $C(w)$ such that
$m(r)\leq C(w) e^{-r^{1-w/p}}$ and we get hence  
$$
E_1 \leq C(w) e^{-R^{p-w}}.
$$
which is the claim of the lemma, since $R_1=r=R^{p}$.

\end{proof}

\begin{proof}[Proof of Theorem \ref{thm: main}]
We recall that $\sigma=\rho_{\Phi}$ with $\Phi$ a ground state of $H+J$.
By the local gap assumption \ref{ass: local gap}, we have 
$$\langle \bar P_{B_r} \rangle_\sigma \leq C r^{d_\gamma}\langle  H_{B_r} \rangle_\sigma $$ and by Lemma \ref{lem: isoperimetry iterated}, we then get, for any $w>0$,
\begin{equation}\label{eq: bound on p bar end}
\langle \bar P_{B_r} \rangle_\sigma \leq \delta=  C(w)r^{d_\gamma} e^{-R^{p-w}} , \qquad  r=R^p
\end{equation}
 From \eqref{eq: bound on p bar end} we get, with $P=P_{B_r}$ and for $R$ large enough, $P\Phi\neq 0$, and setting $\Theta=\frac{P\Phi}{||P\Phi||}$, we get $|\langle \Theta, \Phi\rangle| \geq \sqrt{1-\delta}$. From this one bounds ($|\cdot|_\Gamma$ is the global trace norm)
$$
|\rho_\Phi-\rho_{\Theta}|_\Gamma  \leq 2\sqrt{\delta}
$$  
Since $\Theta \in \Ker(H_{B_{r}})$, the OBC regularity assumption \ref{ass: LTO} allows us to bound
$$
|\sigma-\mu|_{B_{r-k}} \leq |\sigma-\rho_{\Theta}|_{B_{r-k}}+ |\rho_{\Theta}-\mu|_{B_{r-k}}  \leq 2\sqrt{\delta} +  p(r)m(k)
$$
with $p(\cdot)$ a polynomial. 
We choose $k=r/2$.  Then we get, for any $w>0$
$$
|\sigma-\mu|_{B_{R^{p}/2}} \leq C(w)e^{-R^{p-w}}
$$
This yields the statement of Theorem \ref{thm: main}.
\end{proof}

\appendix 

\section{Appendix: Locality properties of CP maps}\label{sec: locality of cp}
We establish some vocabulary that is helpful for the proof of Proposition \ref{lem: stitching invertible}.
\subsection{Completely positive maps}
We consider completely positive (CP) maps acting from $\caB(\caG)$ to $\caB(\caG')$, with $\caG,\caG'$ finite-dimensional Hilbert spaces.  We need identity-preserving CP maps and trace-preserving CP maps, which are dual to each other. Let $\Upsilon: \caB(\caG) \mapsto\caB(\caG')$ be an identity preserving CP map, then the adjoint trace-preserving CP map $\Upsilon^*: \caB(\caG') \mapsto \caB(\caG)$ is defined by
$$
\tr^{(\caG')}(S \Upsilon(O))=  \tr^{(\caG)}(\Upsilon^*(S) O), \qquad O\in\caB(\caG), S\in\caB(\caG').
$$
From the Russo–Dye theorem, it follows that identity preserving CP maps are contracting w.r.t.\ operator norm:   $||\Upsilon(O)||\leq ||O||$ and trace-preserving CP maps are contracting w.r.t\ trace norm $||\Upsilon^*(S)||_1 \leq  ||S||_1 $ with $||S||_1=\tr(\sqrt{SS^*})$, cf.\ the notation in Section \ref{sec: trace norms}.

\subsection{Conditional expectations}
Let us consider a bipartite finite-dimensional Hilbert space $\caG=\caG_a\otimes\caG_b$. For any density matrix $\sigma_b$ on $\caG_b$, we can define the conditional expectation 
$$\bbE_{\sigma_b}: \caB(\caG)\mapsto \caB(\caG):  O \mapsto \tr_{\caG_b}((\ident_a \otimes\sigma_b) O)\otimes \ident_b.
$$
Note that $\bbE_{\sigma_b}$ is an identity-preserving CP map and that its range  is $\caB(\caG_a)\otimes\ident_b$ which is naturally identified with $\caB(\caG_a)$.  Its adjoint  is
\begin{equation}\label{eq: adjoint cond exp}
\bbE^*_{\sigma_b}:  \caB(\caG) \mapsto   \caB(\caG):  S \mapsto \tr_{\caG_b}(S)\otimes\sigma_b.  
\end{equation}
Let $\iota_b$ be the density matrix $\tfrac{1}{\dim(\caG_b)}\ident_b$ on $\caG_b$, i.e.\ $\iota_b$ corresponds to the tracial state on $\caB(\caG_b)$. 
The case where $\caG_b$ is $\caH_Z$ or $\widetilde\caH_Z$ for some $Z\subset\Gamma$, and $\sigma_b=\iota_b$ is of special signficance to us and we denote the corresponding conditional expectation simply by $\bbE_Z$. This is consistent with the usage of the symbol  $\bbE_Z$ in Lemma \ref{lem:splitting2}.

\subsection{Almost locality preserving CP maps}
We consider identity-preserving CP maps $\Upsilon$, acting between (subalgebras of) $\widetilde\caA$, i.e.\ there is a spatial structure. 
We say that $\Upsilon$ preserves almost locality if there exists $m\in\caM$ such that,
\begin{equation} \label{eq: locality property sigma}
||\Upsilon(O)- \bbE_{(X_r)^c}\circ \Upsilon(O)  ||  \leq || O || |X | m(r), \qquad    \forall O \in \caA_X, \forall X \subset \Gamma,
\end{equation}
where $\bbE_{(X_r)^c}$ is the conditional expectation using the tracial state on $(X_r)^c$, as discussed above. 
A useful property is that, if $\Upsilon,\Upsilon'$ preserve almost locality (say, with functions $m_1,m_2$) then so does $\Upsilon'\circ\Upsilon$. Indeed, for $O\in \caA_X$, 
\begin{align}
 \Upsilon'\circ\Upsilon(O) &= \Upsilon'\left( \bbE_{(X_{r/2})^c}\circ\Upsilon(O) +E_1 
 \right) \\   
 &=  \bbE_{(X_{r})^c} \circ \Upsilon' \circ \bbE_{(X_{r/2})^c}\circ \Upsilon(O) +E_2 +\Upsilon'(E_1)  \\ 
 & =\bbE_{(X_{r})^c} \circ\Upsilon'\circ\Upsilon(O) -\bbE_{(X_{r})^c} \circ\Upsilon'(E_1)+ E_2 +\Upsilon'(E_1) 
\end{align}
Using that $\Upsilon,\Upsilon'$ are norm-contracting and almost locality preserving, and that the  image of $\bbE_{(X_{r/2})^c}$ is in $\caA_{X_{r/2}}$, the error terms $E_{1,2}$ are bounded as 
$$||E_1 || \leq || O || |X | m_1(r/2)),\qquad ||E_2 || \leq || O || |X_{r/2}| m_2(r/2)).   $$
Since $|X_{r/2}| \leq C|X|r^d$ and also $\bbE_{(X_{r})^c}$ is norm-contracting, we verify that the overall error term is of the form $||O|| m(r)|X| $, and hence $\Upsilon'\circ\Upsilon$ indeed preserves almost locality. 

If an identity-preserving  CP map $\Upsilon$ preserves almost locality, then its adjoint $\Upsilon^*$,
satisfies the following property, for any density matrix $\rho$, 
\begin{equation}\label{eq: duality relation}
|\Upsilon^*(\rho)-\Upsilon^*\circ \bbE^*_{(X_r)^c}(\rho) |_X \leq |X|m(r)
\end{equation} 
with $\bbE^*_{(X_r)^c}(\rho)=\tr_{(X_r)^c}\rho\otimes \iota_{(X_r)^c}$
This follows from \eqref{eq: locality property sigma} by duality, using  \eqref{eq: adjoint cond exp}.

\subsection{Verification of properties 1) and 2) of definition \ref{def: locality stitch}}
We study the the trace-preserving map $\Sigma=\Sigma_Z$, defined in Proposition \ref{lem: stitching invertible}, via its adjoint $\Sigma^*$, an identity-preserving CP map.
We decompose $\Sigma^*=\Upsilon_1 \circ  \ldots \circ  \Upsilon_6  $ with $\Upsilon_i$ the identity-preserving CP maps given below.  
 \begin{itemize}
\item[$\Upsilon_6$:] $\caA\to \widetilde\caA: O\mapsto O \otimes\aux \ident_{\caA'} $
\item[$\Upsilon_5$:] $\widetilde\caA\to \widetilde\caA: O\mapsto VO V^* $
\item[$\Upsilon_4$:] $\widetilde\caA\to \widetilde\caA_{Z^c}: O\mapsto  \bbE_{ \kappa} (O) $ with the density matrix $\kappa= \tr_{Z^c} (V\mu\otimes\aux\mu'V^*)$ on $\caH_Z$.
\item[$\Upsilon_3$:] $\widetilde\caA_{Z^c}\to \widetilde\caA: O\mapsto  \ident_{Z}\otimes O $
\item[$\Upsilon_2$:] $\widetilde\caA\to \widetilde\caA: O\mapsto V^*O V $
\item[$\Upsilon_1$:] $\widetilde\caA\to \caA: O\mapsto  \bbE_{\mu'}( O) $ with the density matrix $\mu'$ on $\caH'$.
\end{itemize}
Each of the maps $\Upsilon_j$ preserves almost locality, i.e.\ satisfies \eqref{eq: locality property sigma}. This is immediate for $\Upsilon_i$ with $i=1,3,4,6$, whereas for $i=2,5$ it is a consequence of the bound \eqref{eq: lr} in Lemma \ref{lem:splitting2}.  Hence, by duality, we obtain that the locality property \eqref{eq: duality relation} is satisfied for $\Upsilon^*=\Sigma$, the map from Proposition \ref{lem: stitching invertible}. 
We then write, for any density matrices $\rho,\omega$,
\begin{align}
\Sigma(\rho)-\Sigma(\omega)=  \left(\Sigma(\rho)-\Sigma\circ\bbE^{*}_{(X_r)^c}(\rho)\right)- \left(\Sigma(\omega)-\Sigma\circ\bbE^{*}_{(X_r)^c}(\omega)\right) +  \Sigma\circ\bbE^{*}_{(X_r)^c}(\rho-\omega)
\end{align}
The first two terms on the right-hand side are bounded in $|\cdot|_X$-norm by \eqref{eq: duality relation}  with  $\Upsilon^*=\Sigma$.  The last term is bounded as 
$$
|\Sigma\circ\bbE^{*}_{(X_r)^c}(\rho-\omega)|_X \leq | \tr_{(X_r)^c}(\rho-\omega)\otimes \iota_{(X_r)^c}|_X \leq  |\rho-\omega|_{X_r}
$$
where we used in particular contractivity of $\Sigma$. This yields property 2) of definition \ref{def: locality stitch}.  
To check property 1) of definition \ref{def: locality stitch}, we compare the action of $\Upsilon_2,\Upsilon_5$ on $\caA_X$ with $X$ far from $\partial Z$ with the action of the identity map.  By the bound \eqref{eq: locality of vz} in Lemma \ref{lem:splitting2}, we have,
$$
||\Upsilon_j(O)-O || \leq ||O||m(\dist(X,\partial Z))|X|, \qquad O \in \caA_X, \qquad j=2,5.
$$
The verification of property 1) now follows by similar reasoning as for property 2).

\section{Appendix: Locality estimates}\label{sec: locality estimates}
We review the standard propagation bounds that are necessary for the proof of Lemma \ref{lem:splitting2}.  These bounds go back to 
\cite{Lieb:1972ts} but we use the recent formulation in \cite{nachtergaele2019quasi}. 
There are a few differences in our treatment compared to \cite{nachtergaele2019quasi}.  Most importantly, we have an underlying finite graph $\Gamma$ equipped with the graph distance, and we always assume that this graph has a finite spatial dimension $d$ as defined in Section \ref{sec: spatial structure}.

\subsection{Norms on interactions}
In the main text, we defined a family of norms $||\cdot||_m$ with $m\in\caM$. 
In \cite{nachtergaele2019quasi}, as in many preceding works, another family of norms is used, which we introduce now.
Namely, one defines so-called $F$-functions $F$ as functions satisfying
\begin{enumerate}
\item $F:\bbN\to\bbR^+$ is non-increasing.
\item There is  $C_F<\infty$ such that, for all $x,y\in\Gamma$, 
$$\sum_{z\in\Gamma}
F(\dist(x, z))F(\dist(z, y)) \leq C_F F(\dist(x, y)).
$$
\item There is $C'_F$ such that, for each ${x\in\Gamma}$, $ \sum_{y\in\Gamma}F(\dist(x,y)) \leq C_F'$.
\end{enumerate} 
Strictly speaking, the latter two properties are of course empty for a finite graph $\Gamma$, but we have in mind that these properties holds for constants $C_F, C_F'$ that can be chosen uniform in $|\Gamma|$. This is consistent with the conventions used in the main text, as explained in Section \ref{sec: result}. 

Then we can define a family of corresponding norms on interactions.
$$
||| z |||_F=\sup_{x,y\in \Gamma}  \sum_{S \supset \{x,y \}}\frac{||z_S||}{F(\dist(x,y))}
$$
(Note that we use the symbol $|||\cdot |||_F$ to distinghuish these norms from the norms $||\cdot||_m$  in the main text). 
If we work with functions decaying faster than polynomials, one can switch back and forth between the norms $|||\cdot|||_F$ and $||\cdot||_m$, as we state in the next lemma.
In stating the comparison, we write $F\in\caM$ meaning that $F\big|_{\bbN^+}\in\caM$ (The functions $m$ in the main text are defined on $\bbN^+=\{1,2,\ldots\}$ whereas the functions $F$ are defined on $\bbN=\{0,1,2,\ldots\}$.)
\begin{lemma}\label{lem: comparison of norms}
For any $m\in\caM$, there is an $F$-function $F$ in $\caM$ such that 
$$
||| z |||_F \leq || z ||_m
$$
For any $F$-function $F\in\caM$, there is a $m\in\caM$ such that
$$
|| z ||_m \leq ||| z |||_F 
$$
\end{lemma}
To prove this, we need the following observation
\begin{lemma}\label{lem: acro with functions}
Let $f:\bbN^+\to\bbR^+$ be a function such that $\lim_{r\to\infty}e^{r^{\alpha}}f(r)=0$ for any $0<\alpha <1$, then there is an $F$-function $F\in\caM$ such that $f(r)\leq F(r)$ for any $r\in\bbN^+$.
\end{lemma}
\begin{proof}
We say a function $f:\bbN^+\to\bbR^+$ is logarithmically {superadditive} if 
\begin{equation}\label{eq: log subadditive}
f(r_1) f (r_2)\leq  f (r_1+r_2)\qquad r_1,r_2 \in \bbN^+.
\end{equation}
Following \cite{bruckner1960minimal}, we define the following transformation of functions $f:
\bbN^+\to \bbR^+$:
$$\superadd(f)(r)= \sup_{\ell\in\bbN^+} \,\mathop{\sup}\limits_{(r_1,\ldots,r_\ell) \in (\bbN^+)^\ell: \sum_i r_i=r }  \, \prod_{i} f(r_i)
$$
This transformation has the following properties
\begin{enumerate}
\item $\superadd(f)$ is logarithmically superadditive.
\item If $f_1\leq f_2$, then $\superadd(f_1)\leq \superadd(f_2)$. In particular, $f \leq \superadd(f)$.
\item If $f$ is logarithmically superadditive, then $\superadd(f)=f$.  
\end{enumerate}
Let now $f$ satisfy the assumption of the lemma and put $f'(r)=c_0(1+r^{d+2})f(r)$ with $c_0$ such that $f'<1$. Then, we put
$$
F(r)=  \frac{1}{c_0(1+r^{d+2})} \superadd(f'(r)), \qquad r\in \bbN^+,
$$
and $F(0)=F(1)$.
Then $F$ is an $F$-function (such properties are extensively discussed in \cite{nachtergaele2019quasi}) and, by property 2) above, $F(r)\geq f(r)$. 
It remains to show that $F$ has sufficiently rapid decay. Since $f'<1$ and $f'e^{r^{\alpha}}\to 0$, we can find $c_\alpha>0$ such that $f'(r)\leq f_\alpha(r)=e^{-(c_\alpha r)^{\alpha}}$. Since $f_\alpha$, is logarithmically superadditive, we have 
$$
\superadd(f') \leq   \superadd(f_\alpha)=  f_\alpha.
$$
Therefore, $\superadd(f')$ and hence $F(r)$ have indeed sufficiently rapid decay and we conclude that $F\in\caM$.
\end{proof}

\begin{proof}[Proof of Lemma \ref{lem: comparison of norms}]

To get the first claim we observe
\begin{align}
||| z |||_F  & = \sup_{x,y} \frac{1}{F(\dist(x,y))} \sum_{S: S \supset \{x,y\} }  ||z_S|| \\
& \leq \sup_{x,y}  \sum_{S: S \supset \{x,y\} }   \frac{||z_S|| }{F(\diam(S))}  \\
& \leq \sup_{x}  \sum_{S: S \ni x }   \frac{||z_S||}{F(\diam(S))} 
\end{align}
In order to have the last expression dominated by $||z||_m$, we need to choose the $F$-function $F$ such that  $F(r) \geq m(r+1)$.  The possibility of doing this was argued in Lemma  \ref{lem: acro with functions}.
For the second claim, we write
\begin{align*}
|| z ||_m  & =\sup_{x} \sum_{S: S \ni x}  \frac{1}{m(\diam(S)+1)} ||z_S||  \\
 & = \sup_{x} \sum_{k=0} \, \sum_{S: S \ni x, \diam(S)=k} \,  \frac{1}{m(k+1)} ||z_S||  \\
 & \leq \sup_{x} \sum_{k=0} \, \sum_{y,z: \dist(y,x)\leq k,  \dist(y,z)= k } \,   \frac{1}{m(k+1)} \sum_{S: S \ni \{y,z\}}    ||z_S||    \\
  & \leq  \sum_{k=0} \frac{(1+C_{\Gamma}k^d)^2}{m(k+1)} \, \sup_{y,z: \dist(y,z)= k} \,   \sum_{S: S \ni \{y,z\}}   ||z_S||  \\
  & =  \sum_{k=0} \frac{1}{1+k^2}     \frac{1}{f(k+1)}  \,  \sup_{y,z: \dist(y,z)= k}  \, \sum_{S: S \ni \{y,z\}} ||z_S||, \qquad f(k+1)=  \frac{m(k+1)}{(1+k^2) (1+C_{\Gamma}k^d)^2} \\
        & \leq  C_0  \sup_{y,z}     \frac{1}{f(\dist(y,z)+1)}      \sum_{S: S \ni \{y,z\}}  ||z_S||, \qquad C_0=\sum_{k=0} \frac{1}{1+k^2} 
\end{align*}

In order to have the last expression dominated by $|||z|||_F$, we need to choose $m \in \caM$ such that, for all $r\in\bbN^+$,  $f(r)\geq C_0 F(r-1)$, and hence, 
$$
 m(r) \geq g(r)= C_0(1+(r-1)^2) (1+C_{\Gamma}(r-1)^d)^2 F(r-1).
$$
The function $g$ is not necessarily non-increasing, which is remedied by setting
$$
m(r)= \sup_{ r' \geq r}   g(r').
$$ 
We can now check that $m \in\caM$ as it inherits the rapid decay from $F$.
\end{proof}

We call $\mathfrak{B}_F$ the set of interactions with finite $||| \cdot|||_F$-norm and 
$\mathfrak{B}_F(I)$, with $I\subset\bbR$ a (possibly infinite) interval, the set of time-dependent interactions $I\ni s\mapsto z(s)$, with finite norm 
$$
|||z|||_F=\sup_{s\in I}   |||z(s)|||_F
$$
and such that $s\mapsto z_S(s)$ is measurable for all $S\subset\Gamma$. 
Pointwise addition of interactions $(z_1+z_2)_S(s)=(z_1)_S(s)+(z_2)_S(s)$ makes $\mathfrak{B}_F$ and  $\mathfrak{B}_F(I)$ into Banach spaces w.r.t.\ the norm  $|||\cdot|||_F$.

\subsection{Locally generated automorphisms}
Let $z \in \mathfrak{B}_F(I)$ and let $s,t \in I$ with $s\leq t$. We let $\alpha_z(t,s)[\cdot]$ be the dynamics (a family of automorphisms  $\caA\to\caA$ or $\widetilde\caA\to\widetilde\caA$ ) generated by the family $z(\cdot)$ acting from time $s$ to time $t$, in the sense that 
$$
\alpha_z(t,s)[A]=A+i\int_s^t du \,  \alpha_z(u,s)([H_z(u),A]),\qquad H_z(u)= \sum_{S \subset\Gamma} z_S(u).
$$
Since $\Gamma$ is finite, the existence and uniqueness of this dynamics follows from elementary facts on matrix-valued ODE's.  For $s=0$, one often abbreviates
$$
\alpha_z(t)=\alpha_z(t,0)
$$
as we did in the main text.  We first state a version of the Lieb-Robinson bound, using the language introduced in Section \ref{sec: prelim}.
\begin{theorem}[Lieb-Robinson bound] \label{thm: LR}
Let $z\in \mathfrak{B}_F(I)$, and let $s,t \in I$ with $s\leq t$.
Let $A\in\caA_{X_A}, B\in\caA_{X_B} $ with $X_A \cap X_B=\emptyset$. Then 
$$
||[\alpha_z(t,s)[A],B] || \leq  \frac{ ||A || ||B ||}{C_F}  e^{C_F |t-s| |||z|||_F}    \left( \sum_{x\in X_A,y\in X_B} F(\dist(x,y))\right)
$$
\end{theorem}
This is a slightly weakened form of Theorem 3.1 in \cite{nachtergaele2019quasi}, except for two details.  Our graph $\Gamma$ is finite and we allow our time-dependent interaction $z$ to be composed of measurable functions $t\mapsto z_S(t)$ for any $S$. Inspection of the proof in \cite{nachtergaele2019quasi}  confirms that this is all what is needed.

The Lieb-Robinson bound does not directly address the question how to write the evolved observable $\alpha_z(t,s)[A]$ as a sum of local terms.  A natural way to do this is to use the  conditional expectation $\bbE_{X}$ as defined in Section \ref{sec: stitching maps}. If $A\in \caA_X$, then we can write
\begin{equation}\label{eq: decomposition into local}
\alpha_z(t,s)[A]=\sum_{k=0}^{\infty} \Delta_{X_k}(\alpha_z(t,s)[A])
\end{equation}
where the maps $\Delta_{X_k}$ are defined as
$$
\Delta_{X_k} = \begin{cases} \bbE_{X^c} &  k=0\\ \bbE_{(X_k)^c}-\bbE_{(X_{k-1})^c}   & k>0 \end{cases}
$$
and we recall that $X_k=\{\dist(X,\cdot) \leq k\}$. We see that the terms corresponding to $k$ in the sum in \eqref{eq: decomposition into local} are supported on the fattened sets $X_k$. 
The Lieb-Robinson bound can then be used to provide bounds on the right-hand side of \eqref{eq: decomposition into local} via the following standard lemma, whose proof we omit.
\begin{lemma}\label{lem: com and norm}
For any $A\in\caA$, 
$$||A-\bbE_{Z^c}[A]||\leq \sup_{O\in \caA_{Z^c}, ||O||=1} ||[O,A]||.$$
 \end{lemma} 
%

\subsubsection{Evolution of interactions}

Let $g$ be an interaction and recall that we associate to it a Hamiltonian $H_g$
$$
H_g= \sum_{S \subset \Gamma} g_S 
$$
Let now $z=(z(s))_{s\in I}$ be as above, i.e.\ a time-dependent interaction, determining a family of autormorphisms $\alpha_z(t,s)$.  We can use these autormorphisms to obtain a new Hamiltonian by transforming $H_g$ into  $\alpha_z(t,s)[H_g] $.  It is very intuitive that this Hamiltonian can again be expressed as a sum of local terms with good decay properties in the size of the support, i.e.\ we can express the Hamiltonian as the Hamiltonian corresponding to an new interaction that we then view as a time-evolved interaction $\alpha_z(t,s)[g]$. A priori, there are multiple ways to define $\alpha_z(t,s)[g]$ since there are multiple ways to represent an operator as a sum of local terms. For convenience, we follow \cite{nachtergaele2019quasi} and set
$$
(\alpha_z(t,s)[g])_S=\sum_{X\subset \Gamma, k\in \bbN: X_{k}=S}   \Delta_{X_k}( \alpha_z(t,s)[g_X]), 
$$
using the operators $\Delta_{X_k}$ introduced above. 
 This definition has the following convenient properties
\begin{enumerate}
\item $H_{\alpha_z(t,s)[g]}=\alpha_z(t,s)[H_g]$.
\item  If $g$ is anchored in a set $\Gamma_g \subset\Gamma$, then $\alpha_z(t,s)[g]$ is also anchored in $\Gamma_g$. 
\item  If the $F$-functions $F_g,F_z$ are in $\caM$, then there is an $F$-function $F'\in\caM$ depending only on $F_g,F_z$, such that 
\begin{equation} \label{eq: bound evolved inter} ||| \alpha_z(t,s)[g]  |||_{F'}  \leq  e^{C_{F_z} |t-s| |||z|||_{F_z}}  ||| g |||_{F_g}. 
\end{equation} 
\end{enumerate}
The first two properties follow immediately from the definitions. The third property is derived starting from the Lieb-Robinsin bound. For details we refer to \cite{nachtergaele2019quasi}, where one can find explicit choices for $F'$.
%
%
%
%
%

This enables us to define many new time-dependent interactions in Lemma \ref{lem: evolved interactions}.
To set the stage, we first define the "commutator" of two interactions $z_1,z_2$ as
$$ (i[z_1,z_2])_S=\sum_{S_1,S_2: S_1\cup S_2=S, S_1\cap S_2 \neq \emptyset} i[z_{S_1}, z_{S_2} ]  
$$ 
We note that if either one of $z_1,z_2$ is anchored in a set $\Gamma_z$, then the commutator is also anchored in $\Gamma_z$.
Then we can establish
\begin{lemma}\label{lem: evolved interactions}
Let $z_0,z_1,z_2$ be time-dependent interactions in $\mathfrak{B}_{F_{0,1,2}}(I)$, respectively, with $F_{0,1,2}\in\caM$. Then the following interactions are in $\mathfrak{B}_F(I)$, with $F\in\caM$ and $F$ depending only on $F_{0}$ (for item 1)) and on $F_1,F_2$ (for items 2,3),
\begin{enumerate}
\item  $\bar{z}_0(s)=-\alpha_{z_0}(s)[z_0(s)]$  
\item     $i[z_1,z_2](s)$
\item  $  w(s)= {z_2}(s)-{z_1}(s)+ i (\alpha_{z_2}(s))^{-1}\left[\int_0^s \alpha_{{z_2}}(u)\Big\{\big[{z_2}(u) - z_1(u),\alpha_{z_1}(s,u)[z_1(s)]\big]\Big\}du \right] $
\end{enumerate}
The time-dependent interaction $\bar{z}_0$ in item 1) satisfies 
\begin{equation}\label{eq: inverse}
\alpha_{\bar{z}_0}(t,s)= (\alpha_{z_0}(t,s))^{-1}
\end{equation}
The time-dependent interaction $w$ in item 3) satisfies
\begin{equation}\label{eq: combined inverse}
\alpha_w(t,s)=(\alpha_{z_1}(t,s))^{-1} \circ \alpha_{z_2}(t,s)
\end{equation}
\end{lemma}
\begin{proof}[Sketch of proof]
First, one verifies that the newly defined time-dependent interactions satisfy the measurability condition, using that $z_{0,1,2}$ do so.   Then, the bound on $\bar{z}_0$ follows directly from \eqref{eq: bound evolved inter}. The identities \eqref{eq: inverse} and \eqref{eq: combined inverse} follow by direct computation.  The bound on $w$ follows by combining all the previous items, i.e.\ the fact that commutators, time-evolution and inverse time-evolution map $\mathfrak{B}_F(I)$ into $\mathfrak{B}_{F'}(I)$ for certain $F'$. 
\end{proof}

On our way to prove Lemma \ref{lem:splitting2}, we first state the following observation concerning the unitary family $V_Z(s)=V(s)$ introduced in the proof of Proposition \ref{lem: stitching invertible} in the main text.
\begin{lemma}\label{lem: auto form for v}  
There is an $F$-function $F\in\caM$ and a time-dependent interaction $l \in \mathfrak{B}_F([0,1])$ such that 
\begin{equation}\label{eq: evolution equation v}
V(s) O  V^*(s)=\alpha_l(s)[O], \qquad  O \in \widetilde\caA
\end{equation}
Moreover, $l$ is anchored in $\partial Z$.
\end{lemma}
\begin{proof}
We recall the time-dependent interactions $q,\hat q$. These were defined as having a finite $||\cdot||_m$ norm, with $m\in\caM$. By Lemma \ref{lem: comparison of norms}, we find an $F$-function $F_0$ such that $q,\hat q$ have a finite $|||\cdot|||_{F_0}$-norm. 
From the definition of $q,\hat q$, we deduce that 
$$
V^*(s)OV(s)= (\alpha_{{q}}(s))^{-1}\circ \alpha_{\hat{q}}(s)[O]
$$
We set
$$
 \tilde{q}(s)=q(s)-\hat q(s)
$$
and we note that $\tilde{q}(s)$ is anchored in $\partial Z$, for any $s$.
Now we invoke item 3) of Lemma \ref{lem: evolved interactions}, with $z_1=q, z_2=\hat{q}$, to conclude that there is $l\in \mathfrak{B}([0,1])$, given by
$$  l(s)= -\tilde{q}(s)- i (\alpha_{\hat{q}}(s))^{-1}\left[\int_0^s \alpha_{\hat{q}}(u)\Big\{\big[\tilde{q}(u),\alpha_{q}(s,u)[q(s)]\big]\Big\}du \right], $$
such that $V(s) O  V^*(s)= \alpha_l(s)[O]$.
Since commutators and time-evolution preserve the property of being anchored in a set, we conclude that  $l$, just as $\tilde{q}$, is anchored in $\partial Z$.   
\end{proof}

\subsection{Proof of Lemma \ref{lem:splitting2}}
For convenience we restate Lemma \ref{lem:splitting2}:
\begin{lemma}[repetition of Lemma \ref{lem:splitting2}]\label{lem:splitting3 repeat}
For any regions $X,Z$ and $O \in \widetilde\caA_X$, with $V=V_Z$ as in Lemma \ref{lem: auto form for v} above, 
\begin{equation}\label{eq: lr repeat}
|| V^*OV -\bbE_{(X_r)^c}(V^*OV) || \leq ||O|| |X| m(r)
\end{equation}
and
\begin{equation}\label{eq: locality of vz repeat}
||V^*OV-O || \leq  ||O || |X|  m(\dist(X,\partial Z)) 
\end{equation}
The same estimates hold as well if we exchange $V$ and $V^*$.
\end{lemma}
To get the bound \eqref{eq: lr repeat}, we invoke Lemma \ref{lem: auto form for v} and  Lemma \ref{lem: com and norm} to relate conditional expectations to commutators. Then, we use the Lieb-Robinson bound (Theorem \ref{thm: LR}) for the dynamics $\alpha_l$:
\begin{align*}
|| V^*OV -\bbE_{(X_r)^c}(V^*OV) || & \leq \sup_{A \in \widetilde\caA_{(X_r)^c}, ||A||\leq 1 }  || [A,\alpha_l(1)[O] ||  \\[0.5mm]
&\leq  ||O|| e^{C_{F_l}|||z_l|||_{F_l}}\left(\sum_{x\in X, y\in (X_r)^c}  F_l(\dist(x,y))\right)
\end{align*}
The last expression between brackets is estimated as $|X|m(r)$, since $F_l\in\caM$ and the spatial dimension $d$ is finite.

Let us now turn to \eqref{eq: locality of vz repeat}. By the Heisenberg equation,
$$
V^* O V - O = i \int_0^1 \alpha_{l}(u) [ H_l(u), O] d u.
$$
Therefore, we can bound
\begin{align*}
 ||V^* O V - O || & \leq  \int_0^1 d u  \sum_S  || \alpha_{l}(u)   [ l_S(u), O] ||  \\
& =    \int_0^1 du   \sum_{S}  ||[l_S(u),O] ||   \\
& \leq  2|| O || \sup_{u\in [0,1]} \sum_{S: S\cap \partial Z\neq\emptyset, S \cap X \neq\emptyset}  ||l_S(u)   || \\
& \leq  2|| O || \sup_{u\in [0,1]}  \sum_{x\in X, y \in \partial Z }  \sum_{S: S \supset \{x,y \}}  ||l_S(u)   || \\
& \leq 2 || O || \sup_{u\in [0,1]}  |||l(u)|||_{F_l}  \sum_{x\in X, y \in \partial Z }  F_l(\dist(x,y))   \\
& = 2 || O || |||l|||_{F_l}  \sum_{x\in X, y \in \partial Z }  F_l(\dist(x,y)) 
\end{align*}
The third line follows because $O$ is supported in $X$ and $l$ is anchored in $\partial Z$.  The  expression in the last line is dominated by $||O|| |X| m(\dist(X,\partial Z))$ upon using that $F_l\in\caM$, that the spatial dimension $d$ is finite,  and bounding  $|||l|||_{F_l}$ by a constant.

To get the bounds  \eqref{eq: lr repeat} and \eqref{eq: locality of vz repeat} for $VOV^*$ instead of $V^*OV$, we note that 
$$
VOV^*=(\alpha_l(1))^{-1}[O] =  \alpha_{\bar{l}}(1)[O] 
$$
with $\bar{l}$ related to $l$ as in Lemma \ref{lem: evolved interactions}. By Lemma \ref{lem: evolved interactions}, $\bar{l}\in\mathfrak{B}_{F_{\bar{l}}}([0,1])$ for some $F$-function $F_{\bar{l}}\in\caM$. We can therefore copy the arguments given above for  $V^*OV=\alpha_{l}(1)[O]$, replacing $l$ by $\bar{l}$. 

\bibliographystyle{quantum} 

\bibliography{impurity}

\begin{thebibliography}{10}

\bibitem{de2017exponentially}
Wojciech De~Roeck and Marius Sch{\"u}tz.
\newblock ``An exponentially local spectral flow for possibly non-self-adjoint
  perturbations of non-interacting quantum spins, inspired by kam theory''.
\newblock \href{https://dx.doi.org/10.1007/s11005-016-0913-z}{Letters in
  Mathematical Physics {\bf 107}, 505--532}~(2017).

\bibitem{del2020lie}
Simone Del~Vecchio, J{\"u}rg Fr{\"o}hlich, Alessandro Pizzo, and Stefano Rossi.
\newblock ``Lie-schwinger block-diagonalization and gapped quantum chains:
  analyticity of the ground-state energy''.
\newblock \href{https://dx.doi.org/10.1016/j.jfa.2020.108703}{Journal of
  Functional Analysis {\bf 279}, 108703}~(2020).

\bibitem{froehlich2020lie}
Juerg Froehlich and Alessandro Pizzo.
\newblock ``Lie--schwinger block-diagonalization and gapped quantum chains''.
\newblock \href{https://dx.doi.org/10.1007/s00220-019-03613-2}{Communications
  in Mathematical Physics {\bf 375}, 2039--2069}~(2020).

\bibitem{yarotsky2006ground}
DA~Yarotsky.
\newblock ``Ground states in relatively bounded quantum perturbations of
  classical lattice systems''.
\newblock \href{https://dx.doi.org/10.1007/s00220-005-1456-9}{Communications in
  mathematical physics {\bf 261}, 799--819}~(2006).

\bibitem{datta1996low}
Nilanjana Datta, Roberto Fern{\'a}ndez, and J{\"u}rg Fr{\"o}hlich.
\newblock ``Low-temperature phase diagrams of quantum lattice systems. i.
  stability for quantum perturbations of classical systems with finitely-many
  ground states''.
\newblock \href{https://dx.doi.org/10.1007/BF02179651}{Journal of statistical
  physics {\bf 84}, 455--534}~(1996).

\bibitem{borgs1996low}
Christian Borgs, R~Koteck{\`y}, and D~Ueltschi.
\newblock ``Low temperature phase diagrams for quantum perturbations of
  classical spin systems''.
\newblock \href{https://dx.doi.org/10.1007/BF02101010}{Communications in
  mathematical physics {\bf 181}, 409--446}~(1996).

\bibitem{lapa2021stability}
Matthew~F Lapa and Michael Levin.
\newblock ``Stability of ground state degeneracy to long-range
  interactions''~(2021).
\newblock  \href{http://arxiv.org/abs/2107.1139}{arXiv:2107.11396}.

\bibitem{bravyi2010topological}
Sergey Bravyi, Matthew~B Hastings, and Spyridon Michalakis.
\newblock ``Topological quantum order: stability under local perturbations''.
\newblock \href{https://dx.doi.org/10.1063/1.3490195}{Journal of mathematical
  physics {\bf 51}, 093512}~(2010).

\bibitem{michalakis2013stability}
Spyridon Michalakis and Justyna~P Zwolak.
\newblock ``Stability of frustration-free hamiltonians''.
\newblock \href{https://dx.doi.org/10.1007/s00220-013-1762-6}{Communications in
  Mathematical Physics {\bf 322}, 277--302}~(2013).

\bibitem{nachtergaele2020quasi}
Bruno Nachtergaele, Robert Sims, and Amanda Young.
\newblock ``Quasi-locality bounds for quantum lattice systems. part ii.
  perturbations of frustration-free spin models with gapped ground states''.
\newblock In Annales Henri Poincar{\'e}.
\newblock \href{https://dx.doi.org/10.1007/s00023-021-01086-5}{Volume~23, pages
  393--511}.
\newblock Springer~(2022).

\bibitem{nachtergaele2021stability}
Bruno Nachtergaele, Robert Sims, and Amanda Young.
\newblock ``Stability of the bulk gap for frustration-free topologically
  ordered quantum lattice systems''~(2021).
\newblock  \href{http://arxiv.org/abs/2102.0720}{arXiv:2102.07209}.

\bibitem{bachmann2012automorphic}
Sven Bachmann, Spyridon Michalakis, Bruno Nachtergaele, and Robert Sims.
\newblock ``Automorphic equivalence within gapped phases of quantum lattice
  systems''.
\newblock \href{https://dx.doi.org/10.1007/s00220-011-1380-0}{Communications in
  Mathematical Physics {\bf 309}, 835--871}~(2012).

\bibitem{de2015local}
Wojciech De~Roeck and Marius Sch{\"u}tz.
\newblock ``Local perturbations perturb—exponentially--locally''.
\newblock \href{https://dx.doi.org/10.1063/1.4922507}{Journal of Mathematical
  Physics {\bf 56}, 061901}~(2015).

\bibitem{kitaev2006anyons}
Alexei Kitaev.
\newblock ``Anyons in an exactly solved model and beyond''.
\newblock \href{https://dx.doi.org/10.1016/j.aop.2005.10.005}{Annals of Physics
  {\bf 321}, 2--111}~(2006).

\bibitem{kitaev2009topological}
Alexei Kitaev and Chris Laumann.
\newblock ``Topological phases and quantum computation''.
\newblock Exact methods in low-dimensional statistical physics and quantum
  computing, Lecture Notes of the Les Houches Summer SchoolPages
  101--125~(2009).
\newblock
  url:~\href{https://arxiv.org/pdf/0904.2771.pdf}{arxiv.org/pdf/0904.2771.pdf}.

\bibitem{nachtergaele2020dispersive}
Bruno Nachtergaele and Nicholas~E Sherman.
\newblock ``Dispersive toric code model with fusion and defusion''.
\newblock \href{https://dx.doi.org/10.1103/PhysRevB.101.115105}{Physical Review
  B {\bf 101}, 115105}~(2020).

\bibitem{henheik2021local}
Joscha Henheik, Stefan Teufel, and Tom Wessel.
\newblock ``Local stability of ground states in locally gapped and weakly
  interacting quantum spin systems''.
\newblock \href{https://dx.doi.org/10.1007/s11005-021-01494-y}{Letters in
  Mathematical Physics {\bf 112}, 1--12}~(2022).

\bibitem{hastings2007quantum}
Matthew~B Hastings.
\newblock ``Quantum belief propagation: An algorithm for thermal quantum
  systems''.
\newblock \href{https://dx.doi.org/10.1103/PhysRevB.76.201102}{Physical Review
  B {\bf 76}, 201102}~(2007).

\bibitem{kato2019quantum}
Kohtaro Kato and Fernando~GSL Brandao.
\newblock ``Quantum approximate markov chains are thermal''.
\newblock \href{https://dx.doi.org/10.1007/s00220-019-03485-6}{Communications
  in Mathematical Physics {\bf 370}, 117--149}~(2019).

\bibitem{hastings2005quasiadiabatic}
Matthew~B Hastings and Xiao-Gang Wen.
\newblock ``Quasiadiabatic continuation of quantum states: The stability of
  topological ground-state degeneracy and emergent gauge invariance''.
\newblock \href{https://dx.doi.org/10.1103/PhysRevB.72.045141}{Physical review
  b {\bf 72}, 045141}~(2005).

\bibitem{freed2014anomalies}
Daniel~S Freed.
\newblock ``Anomalies and invertible field theories''.
\newblock In Proc. Symp. Pure Math.
\newblock Volume~88, pages 25--46.
\newblock ~(2014).
\newblock
  url:~\href{https://arxiv.org/pdf/1404.7224.pdf}{arxiv.org/pdf/1404.7224.pdf}.

\bibitem{KitaevConf}
A.~Kitaev.
\newblock ``On the classification of short-range entangled states''.
\newblock \url{http://scgp.stonybrook.edu/video_portal/video.php?id=2010}.

\bibitem{gu2009tensor}
Zheng-Cheng Gu and Xiao-Gang Wen.
\newblock ``Tensor-entanglement-filtering renormalization approach and
  symmetry-protected topological order''.
\newblock \href{https://dx.doi.org/10.1103/PhysRevB.80.155131}{Physical Review
  B {\bf 80}, 155131}~(2009).

\bibitem{kapustin2020hall}
Anton Kapustin and Nikita Sopenko.
\newblock ``Hall conductance and the statistics of flux insertions in gapped
  interacting lattice systems''.
\newblock \href{https://dx.doi.org/10.1063/5.0022944}{Journal of Mathematical
  Physics {\bf 61}, 101901}~(2020).

\bibitem{Lieb:1972ts}
E.H. Lieb and D.W. Robinson.
\newblock ``{The finite group velocity of quantum spin systems}''.
\newblock \href{https://dx.doi.org/10.1007/978-3-662-10018-9_25}{Commun. Math.
  Phys. {\bf 28}, 251--257}~(1972).

\bibitem{nachtergaele2019quasi}
Bruno Nachtergaele, Robert Sims, and Amanda Young.
\newblock ``Quasi-locality bounds for quantum lattice systems. i. lieb-robinson
  bounds, quasi-local maps, and spectral flow automorphisms''.
\newblock \href{https://dx.doi.org/10.1063/1.5095769}{Journal of Mathematical
  Physics {\bf 60}, 061101}~(2019).

\bibitem{bruckner1960minimal}
A.~Bruckner.
\newblock ``Minimal superadditive extensions of superadditive functions''.
\newblock Pacific J. Math. {\bf 10}, 1155--1162~(1960).
\newblock
  url:~\href{https://msp.org/pjm/1960/10-4/pjm-v10-n4-s.pdf\#page=51}{msp.org/pjm/1960/10-4/pjm-v10-n4-s.pdf\#page=51}.

\end{thebibliography}

\end{document}